\documentclass[a4paper,UKenglish,cleveref, autoref, dvipsnames]{lipics-v2019}

\usepackage{tikz}
\usetikzlibrary{automata,arrows,shapes}
\usetikzlibrary{decorations.pathreplacing}
\usepackage[]{algorithm2e}
\usepackage[]{xcolor}
\usetikzlibrary{angles} 

\usepackage{amsmath}
\usepackage{amssymb,stmaryrd}
\usepackage{enumitem}
\usepackage{multirow}

\nolinenumbers

\title{Energy mean-payoff games}

\author{Véronique Bruyère}{Université de Mons (UMONS), Belgium}{veronique.bruyere@umons.ac.be}{}{}

\author{Quentin Hautem\footnote{supported by a FRIA fellowship}}{Université de Mons (UMONS), 
Belgium}{q.hautem@gmail.com}{}{}

\author{Mickael Randour\footnote{F.R.S.-FNRS Research Associate}}{Université de Mons (F.R.S.-FNRS \& UMONS), Belgium}{mickael.randour@gmail.com}{}{}

\author{Jean-François Raskin}{Université libre de Bruxelles (ULB), Belgium}{jraskin@ulb.ac.be}{}{}

\authorrunning{V. Bruyère, Q. Hautem, M. Randour, and J.-F. Raskin}

\Copyright{V. Bruyère, Q. Hautem, M. Randour, and J.-F. Raskin}



\category{}

\relatedversion{}

\supplement{}

\funding{Work partially supported by the PDR project \emph{Subgame perfection in graph games} (F.R.S.-FNRS), the grant n°F.4520.18 \emph{ManySynth} (F.R.S.-FNRS), the ARC project \emph{Non-Zero Sum Game Graphs: Applications to Reactive Synthesis and Beyond} (Fédération Wallonie-Bruxelles), the EOS project \emph{Verifying Learning Artificial Intelligence Systems} (F.R.S.-FNRS \& FWO), and the COST Action 16228 \emph{GAMENET} (European
Cooperation in Science and Technology).}

\newcommand{\N}{\mathbb{N}}
\newcommand{\Z}{\mathbb{Z}}
\newcommand{\Q}{\mathbb{Q}}

\newcommand{\Nzero}{\N_0}

\newcommand{\playerOne}{\ensuremath{\mathcal{P}_1} } 
\newcommand{\playerTwo}{\ensuremath{\mathcal{P}_2 } }
\newcommand{\playerI}{\ensuremath{\mathcal{P}_i} }

\newcommand{\Plays}{\mathsf{Plays}}
\newcommand{\rhofactor}[2]{\rho{[{#1},{#2}]}}
\newcommand{\rhoprimfactor}[2]{\rho'{[{#1},{#2}]}}

\newcommand{\Obj}{\Omega}  

\newcommand{\pifactor}[2]{\pi{[{#1},{#2}]}}

\newcommand{\MPs}{\mathsf{\overline{MP}}}
\newcommand{\MPi}{\mathsf{\underline{MP}}}
\newcommand{\EMPi}{\mathsf{E \cap \underline{MP}}}
\newcommand{\EMPs}{\mathsf{E \cap \overline{MP}}}

\newcommand{\energy}{\mathsf{Energy}}

\newcommand{\ffclosed}[1]{first-closed}

\makeatletter    

\newcommand{\unnumberedcaption}%
	{\@dblarg{\@unnumberedcaption\@captype}}

\newcommand{\@unnumberedcaption}{}
\long\def\@unnumberedcaption#1[#2]#3{\par
  \addcontentsline{\csname ext@#1\endcsname}{#1}{%
    \protect\numberline{}{\ignorespaces #2}%
    }%
  \begingroup
    \@parboxrestore
    \normalsize
    \@makeunnumberedcaption{\ignorespaces #3}\par
  \endgroup}

\newcommand{\@makeunnumberedcaption}[1]{%
  \vskip\abovecaptionskip
  \sbox\@tempboxa{#1}%
  \ifdim \wd\@tempboxa >\hsize
    #1\par
  \else
    \global \@minipagefalse
    \hbox to\hsize{\hfil\box\@tempboxa\hfil}%
  \fi
  \vskip\belowcaptionskip}

\@ifundefined{abovecaptionskip}{%
  \newlength{\abovecaptionskip}%
  \setlength{\abovecaptionskip}{10pt}%
}{}
\@ifundefined{belowcaptionskip}{%
  \newlength{\belowcaptionskip}%
  \setlength{\belowcaptionskip}{0pt}%
}{}

\makeatother    

\begin{document}
  \maketitle	

\begin{abstract} 
In this paper, we study one-player and two-player energy mean-payoff games. Energy mean-payoff games are games of infinite duration played on a finite graph with edges labeled by 2-dimensional weight vectors. The objective of the first player (the protagonist) is to satisfy an energy objective on the first dimension and a mean-payoff objective on the second dimension. We show that optimal strategies for the first player may require infinite memory while optimal strategies for the second player (the antagonist) do not require memory. In the one-player case (where only the first player has choices), the problem of deciding who is the winner can be solved in polynomial time while for the two-player case we show co-NP membership and we give effective constructions for the infinite-memory optimal strategies of the protagonist. 
\end{abstract}

\section{Introduction}\label{sec:EMP:intro}

Graph games with $\omega$-regular objectives are a canonical mathematical model to formalize and solve the reactive synthesis problem~\cite{PnueliR89}. Extensions of graph games with quantitative objectives have been considered more recently as a model where, not only the correctness, but also the quality of solutions for the reactive synthesis problem can be formalized and optimized. A large effort has been invested in studying games with various kinds of objectives, see e.g.~\cite{BouyerFLMS08,BrimCDGR11,BruyereHR16a,Chatterjee0RR15,ChatterjeeRR14,ChatterjeeV17,em79,VelnerC0HRR15,ZP96}, see also Chapter~27 of~\cite{BloemCJ18} and the survey~\cite{Bruyere17}.

Two particularly important classes of objectives are \emph{mean-payoff} and \emph{energy} objectives. In a mean-payoff game, the edges of the game graph are labeled with integer weights that model payoffs received by the first player (the protagonist) and paid by the second player (the antagonist) when the edge is taken. The game is played for infinitely many rounds, and the protagonist aims at maximizing the mean value of edges traversed during the game while the antagonist tries to minimize this mean value.  Mean-payoff games have been studied in~\cite{em79} where it is shown that memoryless optimal strategies exist for both players. As a corollary of this result, mean-payoff games can be decided in NP $\cap$ co-NP. While pseudo-polynomial time algorithms for solving mean-payoff games have been developed in~\cite{BrimCDGR11,ZP96} as well as the recent pseudo-quasi-polynomial time algorithm in~\cite{DaviaudJL18}, it is a long standing open question whether or not those games can be solved in polynomial time. Energy games were defined more recently in~\cite{ChakrabartiAHS03}. In an energy game, edges are also labeled with integer weights that represent gains or losses of energy. In such a game, the protagonist tries to build an infinite path for which the total sum of energy in all the prefixes is bounded from below, while the antagonist has the opposite goal. Energy games can also be decided in NP $\cap$ co-NP and it is known that they are \emph{inter-reducible} with mean-payoff games~\cite{BouyerFLMS08}.

\emph{Energy mean-payoff} games that combine an energy and a mean-payoff objectives have not been yet studied. This is the main goal of this paper. This is a challenging problem for several reasons. First, multi-dimensional \emph{homogeneous} extensions of mean-payoff and energy games have been studied in a series of recent contributions~\cite{ChatterjeeRR14,JurdzinskiLS15,Velner15,VelnerC0HRR15}, and those works show that when going from one dimension to several, the close relationship between mean-payoff games and energy games is lost and specific new techniques need to be designed for solving those extensions. Second, \emph{pushdown mean-payoff games} have been studied in~\cite{ChatterjeeV17} and shown to be undecidable. 
Decision problems for energy mean-payoff games can be reduced to decision problems of pushdown mean-payoff games, even to the subclass of pushdown mean-payoff games with a one-letter stack alphabet. Unfortunately, pushdown mean-payoff games are undecidable in general and to the best of our knowledge the one-letter stack alphabet case has not been studied.

\subparagraph{Main contributions.} In this paper, we prove that energy mean-payoff games are decidable and more precisely, their decision problems lie in co-NP (Theorem~\ref{thm:twoPlayer}) for both cases of strict and non-strict inequality in the threshold constraint for the mean-payoff objective. To obtain this result, we first study \emph{one-player} energy mean-payoff games and characterize precisely the game graphs in which $\playerOne$ (the protagonist) can build an infinite path that satisfies the energy mean-payoff objective (Theorem~\ref{thm:caracterisation} and Theorem~\ref{thm:caracterisationLarge}). This characterization leads to polynomial time algorithms to solve the decision problems in the one-player case (Theorem~\ref{thm:onePlayer}). Then we show that in \emph{two-player} energy mean-payoff games memoryless optimal strategies always exist for $\playerTwo$ (the antagonist) who aims at spoiling the energy mean-payoff objective of $\playerOne$ (Proposition~\ref{prop:memoryless}). Combined with the polynomial time algorithms for the one-player case, this result leads to co-NP membership of the decision problems. While the memoryless result for $\playerTwo$ allows us to understand how this player should play in energy mean-payoff games, it does not prescribe how $\playerOne$ should play from winning vertices. To show how to effectively construct optimal strategies for $\playerOne$, we consider a reduction to \emph{4-dimensional energy games} in case of strict inequality for mean-payoff objective (Proposition~\ref{prop:reductionEnergy}). With the result of~\cite{JurdzinskiLS15}, this implies the existence of finite-memory strategies for $\playerOne$ to play optimally and of a pseudo-polynomial time algorithm to solve those instances. For non-strict inequalities, this reduction cannot be applied as, even for the one-player case, infinite-memory strategies are sometimes necessary to play optimally. In this case, we show how we can combine an infinite number of finite-memory strategies, that are played in sequence, in order to play optimally (Proposition~\ref{prop:non-strict-effective}). 

\subparagraph{Related work.} As already mentioned, multi-dimensional conjunctive extensions of mean-payoff games and multi-dimensional conjunctive extensions of energy games have been considered~\cite{ChatterjeeDHR10,ChatterjeeRR14,VelnerC0HRR15}. Deciding the existence of a winning strategy for $\playerOne$ in those games is co-NP-complete. Games with any Boolean combination of mean-payoff objectives have been shown undecidable in~\cite{Velner15}. Games with mean-payoff objectives and $\omega$-regular constraints have been studied in~\cite{ChatterjeeHJ05}, while games with energy objectives and $\omega$-regular constraints have been studied in~\cite{ChatterjeeD12}, and their multi-dimensional extensions in~\cite{AbdullaMSS13,ChatterjeeRR14,ColcombetJLS17}.
In~\cite{JurdzinskiLS15}, the authors have studied multi-dimensional energy games for the fixed initial credit and provided a pseudo-polynomial time algorithm to solve those games when the number of dimensions is fixed. Energy games with bounds on the energy level have been studied in~\cite{FahrenbergJLS11,JuhlLR13}. Games with the combination of an energy objective and an average-energy objective are investigated in~\cite{BouyerHMR017,BouyerMRLL18}. This seemingly related class of games is actually quite different from the energy mean-payoff games studied in this paper: e.g., they are EXPSPACE-hard whereas our games are in co-NP. Infinite-state energy games are investigated in~\cite{AbdullaAHMKT14} where energy objectives are studied on infinite game structures induced by one-counter automata or pushdown automata. Some work on other models dealing with energy have been studied, as battery edge systems~\cite{BokerHR14} and consumption games~\cite{BrazdilCKN12}. In the latter games, minimization of running costs have also been investigated~\cite{BrazdilKKN14}. Optimizing the expected mean-payoff in energy MDP's have been studied in~\cite{BrazdilKN16}. In~\cite{Kucera12}, Kucera presents an overview of results related to games and counter automata, which are close to energy constraints.

We now discuss \emph{mean-payoff pushdown games}~\cite{ChatterjeeV17} in more details. In those games, a stack is associated with a finite game structure, and players move from vertex to vertex while applying operations on the stack. Those operations are \emph{push} a letter, \emph{pop} a letter or \emph{skip} and can be respectively represented with weights $1$, $-1$ and $0$. The authors show that one-player pushdown games can be solved in polynomial time, thanks to the existence of \emph{pumpable paths}. Moreover, already in this case, $\playerOne$ needs infinite memory to win in mean-payoff pushdown games. In the two-player setting, determining the winner is undecidable. Doing a straight reduction of one-player energy mean-payoff games to one-player mean-payoff pushdown games would lead to a pseudo-polynomial solution, whereas we show here that we can solve the former games in polynomial time. In addition, we cannot use the concept of pumpable paths to obtain those results as the construction of~\cite{ChatterjeeV17} is inherent to the behavior of the stack of mean-payoff pushdown games. Indeed, after one step, the height of the stack can only change of one unity ($+1,-1,0$), whereas in energy mean-payoff games, the energy level can vary from $-W$ to $+W$, for an arbitrarily large integer $W \in \mathbb{N}$.

\subparagraph{Structure of the paper.} In Sect.~2, we introduce the necessary notations and preliminaries to this work. In Sect.~3, we study the one-player energy mean-payoff games. In Sect.~4, we study the two-player energy mean-payoff games.

\section{Preliminaries} \label{sec:prelim}

In this section, we introduce energy mean-payoff games and the related decision problems studied in this paper.

\subparagraph{Games structures.} 
A \emph{game structure} is a weighted directed graph $G = (V, V_1,V_2, E, w)$ such that $V_1, V_2$ form a partition of the finite set $V$, $V_i$ is the set of vertices controlled by player~\playerI, $i \in \{1,2\}$, $E \subseteq V \times V$ is the set of edges such that for all $v \in V$, there exists $v' \in V$ such that $(v,v') \in E$, and $w = (w_1,w_2) : E \rightarrow \Z^2$ is a weight function that assigns a pair of weights $w(e) = (w_1(e),w_2(e))$ to each edge $e \in E$. In the whole paper, we denote by $|V|$ the number of vertices of $V$, by $|E|$ the number of edges of $E$, and by $||E|| \in \Nzero$ the largest absolute value used by the weight function $w$. We say that a game structure is a \emph{player-$i$ game structure} when player~\playerI controls all the vertices, that is, $V_i = V$.

A \emph{play} in $G$ from an \emph{initial vertex} $v_0$ is an infinite sequence $\rho = \rho_0\rho_1 \ldots \rho_k \ldots$ of vertices such that $\rho_0 = v_0$ and $(\rho_k,\rho_{k+1}) \in E$ for all $k \geq 0$. A \emph{factor} of $\rho$, denoted by $\rhofactor{k}{\ell}$, is the finite sequence $\rho_k\rho_{k+1} \ldots \rho_\ell$. When $k = 0$, we say that $\rhofactor{0}{\ell}$ is the \emph{prefix} of length $\ell$ of $\rho$. The \emph{suffix} $\rho_k\rho_{k+1} \ldots $ of $\rho$ is denoted by $\rhofactor{k}{\infty}$. The set of plays in $G$ is denoted by $\Plays(G)$ or simply $\Plays$. A path or a cycle is \emph{simple} if there are no two occurrences of the same vertex (except for the first and last vertices in the cycle). A \emph{multicycle} $\cal C$ is a multiset of simple cycles (that may or may not be connected to each other). We extend the weight function $w$ to paths (resp. cycles, multicycles) as the sum $w(\pi) = (w_1(\pi),w_2(\pi))$ of the weights of their edges. In particular, for a multicycle $\cal C$, we have $w({\cal C}) = \sum_{\pi \in {\cal C}} w(\pi)$.

Let us recall the following notion. Given a path $\pi = \pi_0 \pi_1 \cdots \pi_n$, we consider its \emph{cycle decomposition} into a \emph{multiset of simple cycles} as follows. We push successively vertices $\pi _0, \pi_1, \ldots$ onto a stack. Whenever we push a vertex $\pi_\ell$ equal to a vertex $\pi_k$ already in the stack, i.e. a simple cycle $C = \pi_k \cdots \pi_{\ell}$ is formed, we remove this cycle from the stack except $\pi_k$ (we remove all the vertices until reaching $\pi_k$ that we let in the stack) and add $C$ to the cycle decomposition multiset of $\pi$. The cycle decomposition of a play $\rho = \rho_0\rho_1 \ldots $ is defined similarly.

For each dimension $j \in \{1,2\}$, the weight or \emph{energy level} of the prefix $\rhofactor{0}{k}$ of a play $\rho$ is $w_j(\rhofactor{0}{k})$, and the \emph{mean-payoff-inf} (resp. \emph{mean-payoff-sup}) of $\rho$ is $\MPi_j(\rho) = \liminf_{k \rightarrow \infty} \frac{1}{k} \cdot w_j(\rhofactor{0}{k})$ (resp. $\MPs_j(\rho) = \limsup_{k \rightarrow \infty}  \frac{1}{k} \cdot w_j(\rhofactor{0}{k})$). The following properties hold for both mean-payoff values. First, they are prefix-independent, that is, $\MPi_j(\pi \rho) = \MPi_j(\rho)$ and $\MPs_j(\pi \rho) = \MPs_j(\rho)$ for all finite paths $\pi$. Second for a play $\rho = \rho_0 \ldots \rho_{k-1}(\rho_k \dots \rho_l)^\omega$ that is eventually periodic, its mean-payoff-inf and mean-payoff-sup values coincide and are both equal to the average weight of the cycle $\rho_k \dots \rho_l\rho_k$, that is, $\frac{1}{l-k+1} \cdot w_j(\rho_k \dots \rho_l\rho_k)$.

\subparagraph{Strategies.} 
Given a game structure $G$, a \emph{strategy} $\sigma_i$ for player~$\playerI$ is a function $V^* \cdot V_i \to V$ that assigns to each path $\pi v$ ending in a vertex $v \in V_i$ a vertex $v'$ such that $(v, v') \in E$. Such a strategy $\sigma_i$ is \emph{memoryless} if it only depends on the last vertex of the path, i.e. $\sigma_i(\pi v) = \sigma_i(\pi' v)$ for all $\pi v, \pi' v \in V^* \cdot V_i$. It is a \emph{finite-memory} strategy if it can be encoded by a deterministic \emph{Moore machine} ${\cal M} = (M, m_0, \alpha_U, \alpha_N)$ where $M$ is a finite set of states (the memory of the strategy), $m_0 \in M$ is an initial memory state, $\alpha_U : M \times V \rightarrow M$ is an update function, and $\alpha_N : M \times V_i \rightarrow V$ is a next-move function. Such a machine defines a strategy $\sigma_i$ such that $\sigma_i(\pi v) = \alpha_N(\widehat{\alpha}_U(m_0,\pi),v)$ for all paths $\pi v \in V^* \cdot V_i$, where $\widehat{\alpha}_U$ extends $\alpha_U$ to paths as expected. The \emph{memory size} of $\sigma_i$ is then the size $|M|$ of $\cal M$. In particular $\sigma_i$ is memoryless when it has memory size one.

Given a strategy $\sigma_i$ for $\playerI$, a play $\rho$ is \emph{consistent} with $\sigma_i$ if for all its prefixes $\rhofactor{0}{k} \in V^* \cdot V_i$, we have $\rho_{k+1} = \sigma_i(\rhofactor{0}{k})$. A finite path $\pi$ consistent with $\sigma_i$ is defined similarly. Given a finite-memory strategy $\sigma_i$ and its Moore machine ${\cal M}$, we denote by $G(\sigma_i)$ the game structure obtained as the product of $G$ with $\cal M$. Notice that the set of plays from an initial vertex $v_0$ that are consistent with $\sigma_i$ is then exactly the set of plays in $G(\sigma_i)$ starting from $(v_0,m_0)$ where $m_0$ is the initial memory state of $\cal M$.

\subparagraph{Objectives.} 
Given a game structure $G$ and an initial vertex $v_0$, an \emph{objective} for player~$\playerOne$ is a set of plays $\Obj \subseteq \Plays(G)$. Given a strategy $\sigma_1$ for $\playerOne$, we say that $\sigma_1$ is \emph{winning for $\playerOne$ from $v_0$} if all plays $\rho \in \Plays(G)$ from $v_0$ that are consistent with $\sigma_1$ satisfy $\rho \in \Obj$. Given a strategy $\sigma_2$ for $\playerTwo$, we say that $\sigma_2$ is \emph{winning for $\playerTwo$ from $v_0$} if all plays $\rho \in \Plays(G)$ from $v_0$ that are consistent with $\sigma_2$ satisfy $\rho \not\in \Obj$.

We here consider the following objectives for dimension $j \in \{1,2\}$:
\begin{itemize}
\item \emph{Energy objective}. 
Given an \emph{initial credit} $c_0 \in \N$, the objective $\energy_j(c_0) = \{\rho \in \Plays(G) \mid \forall k \geq 0, c_0 + w_j(\rhofactor{0}{k}) \geq 0 \}$ requires that the energy level remains always nonnegative in dimension $j$.  
\item \emph{Mean-payoff-inf objective}. 
The objective $\MPi_j(\sim 0) = \{\rho \in \Plays(G) \mid \MPi_j(\rho) \sim 0 \}$ with ${\sim} \in \{>,\geq\}$ requires that the mean-payoff-inf value is $\sim 0$ in dimension $j$.
\item \emph{Mean-payoff-sup objective}.
The objective $\MPs_j(\sim 0) = \{\rho \in \Plays(G) \mid \MPs_j(\rho) \sim 0 \}$ with ${\sim} \in \{>,\geq\}$ requires that the mean-payoff-sup value is $\sim 0$ in dimension $j$.
\end{itemize}

\begin{remark} \label{rem:threshold}
Notice that it is not a restriction to work with threshold $0$ in mean-payoff-inf/sup objectives. Indeed arbitrary thresholds $\frac{a}{b} \in \Q$ can be reduced to threshold $0$ by replacing the weight function $w$ of $G$ by the function $b \cdot w - a$. Notice also that it is not a restriction to work with integer weights $w_1(e),w_2(e)$ labeling each edge $e \in E$. Indeed, as we work with threshold $0$, an arbitrary weight function $w : E \rightarrow \Q^2$ can be replaced by the function $b \cdot w : E \rightarrow \Z^2$ with an appropriate $b \in \Nzero$.
\end{remark}

\subparagraph{Decision problems.} 
In this paper we consider the following four variants of a decision problem implying an energy objective on the first dimension and a mean-payoff objective on the second dimension. Let ${\sim} \in \{>,\geq\}$:

\begin{itemize}
\item The \emph{energy mean-payoff decision problem $\EMPi^{\sim 0}$} asks, given a game structure $G$ and an initial vertex $v_0$, to decide whether there exist an initial credit $c_0 \in \N$ and a winning strategy $\sigma_1$ for player~$\playerOne$ from $v_0$ for the objective~$\Obj = \energy_1(c_0) \cap \MPi_2(\sim 0)$.
\item The \emph{energy mean-payoff decision problem $\EMPs^{\sim 0}$} asks, given a game structure $G$ and an initial vertex $v_0$, to decide whether there exist an initial credit $c_0 \in \N$ and a winning strategy $\sigma_1$ for player~$\playerOne$ from $v_0$ for the objective~$\Obj = \energy_1(c_0) \cap \MPs_2(\sim 0)$.
\end{itemize}

In this context, we also use the terminology of \emph{energy mean-payoff objectives} or \emph{energy mean-payoff games}. 


\begin{figure}[h]
\begin{minipage}[c]{.45\linewidth}
\centering
  \begin{tikzpicture}[scale=4]
    \everymath{\scriptstyle}
    \draw (0,0) node [circle, draw] (A) {$v_0$};
    \draw (0.75,0) node [circle, draw] (B) {$v_1$};
    
	\draw[->,>=latex] (A) to[bend left] node[above,midway] {$(0,-1)$} (B);
	\draw[->,>=latex] (B) to[bend left] node[below,midway] {$(0,-1)$} (A);
	
    \draw[->,>=latex] (B) .. controls +(45:0.4cm) and +(135:0.4cm) .. (B) node[above,midway] {$(-1,3)$};
    \draw[->,>=latex] (A) .. controls +(45:0.4cm) and +(135:0.4cm) .. (A) node[above,midway] {$(1,-1)$};
	\path (-0.2,0) edge [->,>=latex] (A);    
    
    \end{tikzpicture}
\caption{Energy mean-payoff game where $\playerOne$ wins with finite-memory for problems $\EMPi^{> 0}$ and $\EMPs^{> 0}$.}
\label{fig::EMP:finiteMemory}
\end{minipage}
\hspace{.4cm}
\begin{minipage}[c]{.45\linewidth}
\centering
  \begin{tikzpicture}[scale=4]
    \everymath{\scriptstyle}
    \draw (0,0) node [circle, draw] (A) {$v_0$};
    \draw (0.75,0) node [circle, draw] (B) {$v_1$};
    
	\draw[->,>=latex] (A) to[bend left] node[above,midway] {$(0,-1)$} (B);
	\draw[->,>=latex] (B) to[bend left] node[below,midway] {$(0,-1)$} (A);
	
    \draw[->,>=latex] (B) .. controls +(45:0.4cm) and +(135:0.4cm) .. (B) node[above,midway] {$(-1,1)$};
    \draw[->,>=latex] (A) .. controls +(45:0.4cm) and +(135:0.4cm) .. (A) node[above,midway] {$(1,-1)$};
	\path (-0.2,0) edge [->,>=latex] (A);    
    
    \end{tikzpicture}
\caption{Energy mean-payoff game where $\playerOne$ needs infinite memory to win for problems $\EMPi^{\geq 0}$ and $\EMPs^{\geq 0}$.}
\label{fig::EMP:infiniteMemory}
\end{minipage}
\end{figure}

\subparagraph{Introductory examples.} We provide two examples to illustrate the introduced concepts. 

\begin{example} \label{ex:memfinie}
Consider the player-1 game structure $G$ depicted in Figure~\ref{fig::EMP:finiteMemory}. Consider the cycle $C = v_0v_0v_0v_1v_1v_1v_0$ that loops twice on $v_0$, goes to $v_1$, loops twice on $v_1$, and comes back to $v_0$. Observe that $w(C) = (w_1(C),w_2(C)) = (0,2)$. Hence $\playerOne$ has a winning strategy, that consists in looping forever in this cycle $C$, for all four variants of the energy mean-payoff decision problem.
\end{example}

The second example will be useful later in this article. A similar example is given in~\cite{ChatterjeeV17}.

\begin{example} \label{ex:meminfinie} 
Consider the player-1 game structure $G$ depicted on Figure~\ref{fig::EMP:infiniteMemory}. It differs from the game structure of Figure~\ref{fig::EMP:finiteMemory} only by the weight $(-1,1)$ (instead of $(-1,3)$) of the edge $(v_1,v_1)$. We are going to show that $\playerOne$ has a winning strategy for both problems $\EMPi^{\geq 0}$ and $\EMPs^{\geq 0}$ with initial credit $c_0 = 0$. Notice that when there is only one player, the existence of a winning strategy for $\playerOne$ in the energy mean-payoff decision problem  is equivalent to the existence of a play belonging to the energy mean-payoff objective $\energy_1(c_0) \cap \MPi_2(\sim 0)$ (or $\energy_1(c_0) \cap \MPs_2(\sim 0)$) for some $c_0$. 

First, we show that the answer to both problems $\EMPi^{\geq 0}$ and $\EMPs^{\geq 0}$ is {\sf No} if $\playerOne$ only uses \emph{finite-memory} strategies. Indeed, any finite-memory strategy induces an outcome $\rho$ that eventually loops in some cycle $C$ of $G$. Let $C_0$ (resp. $C_1$, $C_2$) be the simple cycle $(v_0,v_0)$ (resp. $(v_1,v_1)$, $(v_0,v_1,v_0)$). If $\rho$ eventually loops forever on either cycle $C_0$ or cycle $C_1$, then clearly the mean-payoff objective or the energy objective is not satisfied. Therefore, cycle $C_2$ has to be taken in $C$ and we can assume that $C_0$ is visited $\alpha \in \N$ times, $C_1$ is visited $\beta \in \N$ times and $C_2$ is visited $\gamma \in \Nzero$ times along $C$. From equation $w(C) = \alpha \cdot w(C_0) + \beta \cdot w(C_1) + \gamma \cdot w(C_2) = (\alpha - \beta, -\alpha+\beta-2\cdot \gamma)$, we need to have $\alpha - \beta \geq 0$ for the energy objective. Indeed, if $\alpha - \beta < 0$, then for all initial credits $c_0$, the energy of $\rho$ will eventually drop below $0$. We also need to have $-\alpha+\beta-2\cdot \gamma \geq 0$ for the mean-payoff objective as $\MPi_2(\rho) = \MPs_2(\rho)$ is equal to the average weight of the cycle $C$. However as $\alpha - \beta \geq 0$ and $\gamma > 0$, then $-\alpha+\beta-2\cdot \gamma < 0$. This shows that $\playerOne$ cannot win under finite-memory strategies.

Let us now show that with \emph{infinite-memory} strategies, the answer to both problems $\EMPi^{\geq 0}$ and $\EMPs^{\geq 0}$ is {\sf Yes}. Let us first indicate how $\playerOne$ can win for the objective $\Obj = \energy_1(c_0) \cap \MPi_2(\geq 0)$ with $c_0 =0$. Consider the following strategy $\sigma_1$ for $\playerOne$: 
\begin{enumerate}
\item Initialize $Z = 1$
\item At round $Z$
\begin{enumerate}
\item Loop $Z$ times in cycle $C_0$
\item Take edge $(v_0,v_1)$
\item Loop $Z$ times in cycle $C_1$
\item Take edge $(v_1,v_0)$
\end{enumerate}
\item Increment $Z$ by $1$ and goto 2.
\end{enumerate}
Let us show that  $\rho \in \Obj$ where $\rho$ is the play from $v_0$ consistent with $\sigma_1$. Clearly, the energy level on the first dimension never drops below zero by construction, thus we only focus on the second dimension. Intuitively, the mean-payoff-inf value of $\rho$ will be nonnegative since the average weight at round $Z$ is of the form $\frac{-Z}{Z^2}$ which converges to $0$. Let us explain why in more details. Consider any prefix $\pi = \rhofactor{0}{k}$ of $\rho$ during round $Z$. Then one can check that $k \leq \sum_{l=1}^{Z}(2 \cdot l+2) = (Z-1)(Z+2)$ (each complete round $l$ uses $2 \cdot l+2$ edges) and that $w_2(\pi) \geq -3 \cdot Z + 1$ which is the energy level just after performing (b) (since afterwards the sum of weights grows when looping $Z$ times in $C_1$). Thus during round $Z$, $\frac{1}{k} \cdot w_2(\rhofactor{0}{k}) \geq \frac{-3 \cdot Z+1}{(Z-1)(Z+2)}$ which converges to 0 when $Z \rightarrow \infty$. It follows that $\MPi_2(\rho) \geq 0$. This shows that $\rho \in \Obj$.

Finally, notice that the above argument also holds for $\Obj = \energy_1(c_0) \cap \MPs_2(\geq 0)$ with $c_0 = 0$ since $\MPs_2(\rho) \geq \MPi_2(\rho)$ for any play $\rho$. 
\end{example}

\section{One-player setting} \label{sec:OneP}

Within this section, we investigate player-1 game structures, that is, game structures where player~$\playerOne$ is the only one to play. In this context, $\playerOne$ has a winning strategy for the energy mean-payoff objective for some initial credit $c_0$  if and only if there exists a play belonging to this objective. For player-1 game structures, we show that the energy mean-payoff decision problem can be solved in polynomial time for all of its four variants. However depending on the used relation ${\sim} \in \{>,\geq\}$ for the mean-payoff objective, memory requirements for winning strategies of $\playerOne$ differ. We already know that $\playerOne$ needs infinite memory in case of non-strict inequalities by Example~\ref{ex:meminfinie}. In case of strict inequalities, we show that finite-memory strategies are sufficient for $\playerOne$. All these results will be useful in Section~\ref{sec:twoP} when we will investigate the general case of two-player energy mean-payoff games.

\begin{theorem} \label{thm:onePlayer}
The energy mean-payoff decision problem for player-1 game structures can be solved in polynomial time. Moreover,
\begin{itemize}
\item for both problems  $\EMPi^{> 0}$ and $\EMPs^{> 0}$, pseudo-polynomial-memory strategies are sufficient and necessary for $\playerOne$ to win;
\item for both problems $\EMPi^{\geq 0}$ and $\EMPs^{\geq 0}$, in general, $\playerOne$ needs infinite memory to win.
\end{itemize}
\end{theorem}

To prove Theorem~\ref{thm:onePlayer}, we will characterize the existence of a winning strategy for $\playerOne$ for some initial credit $c_0$ by the existence of a particular cycle or multicycle, that we call \emph{good}.

\begin{definition}\label{def:goodCycle}
Let $G$ be a game structure and $v_0$ be an initial vertex. 
\begin{itemize}
\item We say that a cycle~$C$  is a \emph{good cycle} if $w_1(C) \geq 0$ and $w_2(C) > 0$. A good cycle $C$ is \emph{reachable} if it is reachable from $v_0$. 
\item We say that a multicycle $\cal C$ is a \emph{good multicycle} if $w_1({\cal C}) \geq 0$ and $w_2({\cal C}) \geq 0$. A good multicyle $\cal C$ is \emph{reachable} if all its simple cycles are in the same connected component reachable from~$v_0$. 
\end{itemize}
\end{definition}

There exists a simple characterization of the existence of a winning strategy for $\playerOne$ for either the objective $\energy_1(c_0) \cap \MPi_2(> 0)$ or the objective $\energy_1(c_0) \cap \MPs_2(> 0)$ for some initial credit $c_0$: both are equivalent to the existence of a reachable good cycle. 

\begin{theorem}\label{thm:caracterisation}
Let $G$ be a player-1 game structure and $v_0$ be an initial vertex. The following assertions are equivalent.
\begin{enumerate}
\item There exist an initial credit $c_0$ and a winning strategy for $\playerOne$ from $v_0$ for the objective $\energy_1(c_0) \cap \MPi_2(> 0)$.
\item There exist an initial credit $c_0$ and a winning strategy for $\playerOne$ from $v_0$ for the objective $\energy_1(c_0) \cap \MPs_2(> 0)$.
\item There exists a reachable good cycle.
\end{enumerate}
\end{theorem}

In case of non-strict inequalities, there exists also a simple characterization: $\playerOne$ can win for either the objective $\energy_1(c_0) \cap \MPi_2(\geq 0)$ or the objective $ \energy_1(c_0) \cap \MPs_2(\geq 0)$ for some initial credit $c_0$ if and only if there exists a reachable good multicycle. 

\begin{theorem}\label{thm:caracterisationLarge}
Let $G$ be a player-1 game structure and $v_0$ be an initial vertex. The following assertions are equivalent.
\begin{enumerate}
\item There exist an initial credit $c_0$ and a winning strategy for $\playerOne$ from $v_0$ for the objective $\energy_1(c_0) \cap \MPi_2(\geq 0)$.
\item There exist an initial credit $c_0$ and a winning strategy for $\playerOne$ from $v_0$ for the objective $\energy_1(c_0) \cap \MPs_2(\geq 0)$.
\item There exists a reachable good multicycle.
\end{enumerate}
\end{theorem}

A similar characterization appears for multi-mean-payoff games and multi-energy games studied in~\cite{VelnerC0HRR15}: when the objective is an intersection of several mean-payoff-inf objectives (resp. several energy objectives), and when he plays alone, $\playerOne$ has a winning strategy if and only if there exists a reachable non negative multicycle (resp. a reachable non negative cycle) in the game structure. Nevertheless, the proofs of those results differ substantially from the proofs of our results.

Let us illustrate the statements of Theorems~\ref{thm:caracterisation} and~\ref{thm:caracterisationLarge} with the two previous Examples~\ref{fig::EMP:finiteMemory} and~\ref{fig::EMP:infiniteMemory}.

\begin{example}
We first come back to the game structure of Figure~\ref{fig::EMP:finiteMemory}. The cycle $C$ mentioned in Example~\ref{ex:memfinie} is a reachable good cycle since $w(C) = (0,2)$. 
    By Theorem~\ref{thm:caracterisation}, it follows that $\playerOne$ is winning for the energy mean-payoff decision problem with strict inequalities (and thus also with non-strict inequalities), as already observed in Example~\ref{ex:memfinie}.

Let us now come back to the player-1 game structure of Figure~\ref{fig::EMP:infiniteMemory}. Recall that there exists an infinite-memory winning strategy for $\playerOne$ for $c_0 = 0$ in case of non-strict inequalities but no finite-memory winning strategy for any $c_0$. By Theorem~\ref{thm:caracterisationLarge}, there should exist a reachable good multicycle. Indeed, consider the multicycle ${\cal C} = \{C,C'\}$ with $C = (v_0,v_0)$ and $C' = (v_1,v_1)$: we have $w({\cal C}) = w(C) + w(C') = (1,-1) + (-1,1) = (0,0)$. Moreover by Theorems~\ref{thm:onePlayer} and~\ref{thm:caracterisation}, there is no reachable good cycle in this game.
\end{example}

The rest of Section~\ref{sec:OneP} is devoted to the proofs of the above mentioned results. This needs several intermediate steps that are detailed below.

\subsection{Characterization in case of strict inequalities}

We begin by providing the proof for the characterization stated in Theorem~\ref{thm:caracterisation}.
Let us first give a definition related to the energy objective.

\begin{definition}\label{def:localMin}
Let $G$ be a game structure and $\rho \in \Plays(G)$ be a play. We say that position $k$ in $\rho$ is a \emph{local minimum} for the energy if $\forall \ell > k$, $w_1(\rhofactor{0}{k}) \leq w_1(\rhofactor{0}{\ell})$.
\end{definition}

Thus, a position $k$ is a local minimum for the energy if from this position, the energy never drops below $w_1(\rhofactor{0}{k})$. In other words, we have $w_1(\rhofactor{k}{\ell}) \geq 0$ for all $\ell > k$. We show that if a play satisfies the energy objective, then necessarily there are infinitely many local minima for the energy in this play.

\begin{lemma}\label{lem:minLocaux}
If a play $\rho$ satisfies an energy objective $\energy_1(c_0)$ with $c_0 \in \N$, then there are infinitely many local minima for the energy in $\rho$.
\end{lemma}

\begin{proof}
As $\forall k$, $w_1(\rhofactor{0}{k}) + c_0 \geq 0$, there exists $k_0$ such that $w_1(\rhofactor{0}{k_0}) \leq w_1(\rhofactor{0}{\ell})$ for all $\ell \geq k_0$. Now, as for all $\ell \geq k_0$, $w_1(\rhofactor{0}{k_0}) \leq w_1(\pifactor{0}{\ell})$, there exists $k_1 > k_0$ such that $w_1(\rhofactor{0}{k_1}) \leq w_1(\rhofactor{0}{\ell})$ for all $\ell \geq k_1$. We continue this construction to build a sequence of indexes $(k_n)_{n \geq 0}$ such that each $k_n$ is a local minimum for the energy by construction.
\end{proof}

The next lemma provides a partial proof of Theorem~\ref{thm:caracterisation}. It states that given a reachable good cycle, for a well-chosen initial credit, there exists a winning strategy for $\playerOne$ for problems $\EMPi^{> 0}$ and $\EMPs^{> 0}$ (that consists in reaching the simple good cycle and looping in it). This lemma also states that when the reachable cycle has a weight $\geq (0,0)$, then $\playerOne$ has a winning strategy for problems $\EMPi^{\geq 0}$ and $\EMPs^{\geq 0}$.

\begin{lemma}\label{lem:cycle}
Let $G$ be a game structure and $v_0$ be an initial vertex. 
\begin{itemize}
    \item If $G$ has a reachable good cycle $C$, then there exist an initial credit $c_0$ and a winning strategy for $\playerOne$ from $v_0$ for the objective $\energy_1(c_0) \cap \MPi_2(> 0)$ (resp. $\energy_1(c_0) \cap \MPs_2(> 0)$).
    \item If $G$ has a reachable cycle $C$ such that $w(C) \geq (0,0)$, then there exist an initial credit $c_0$ and a winning strategy for $\playerOne$ from $v_0$ for the objective $\energy_1(c_0) \cap \MPi_2(\geq 0)$ (resp. $\energy_1(c_0) \cap \MPs_2(\geq 0)$).
\end{itemize}
Moreover, when this cycle is simple, the winning strategy is memoryless. 
\end{lemma}

\begin{proof}
The first case is easy to prove with $c_0 = (|V|-1) \cdot ||E||$. Indeed, consider a reachable good cycle $\pi = \pi_0 \ldots \pi_k$ and let $v$ be the vertex of $\pi$ where the energy level on the first dimension is the lowest, i.e.,  $v = \pi_j$ where $w_1(\pifactor{0}{j}]) \leq w_1(\pifactor{0}{\ell})$ for all $\ell \in \{0,\ldots,k\}$. Let $v \pi'$ be the good cycle $\pi$ starting from $v$ and $\lambda$ be a simple path from $v_0$ to $v$. As $w_1(v\pi') = w_1(\pi) \geq 0$ and $|w_1(\lambda)| \leq c_0$, the play $\rho = \lambda \pi'^{\omega}$ belongs to $\energy_1(c_0)$. Moreover $\MPi_2(\rho)$ is the average weight (on the second dimension) of the cycle $\pi$ and is thus equal to $w_2(\pi) >0$. Hence we have both $\rho \in \energy_1(c_0) \cap \MPi_2(> 0)$ and $\rho \in \energy_1(c_0) \cap \MPs_2(> 0)$. Notice that if the good cycle is simple, then $\rho$ is the outcome of a memoryless strategy.

The second case is solved similarly.
\end{proof}

We now have all the ingredients to prove Theorem~\ref{thm:caracterisation}.

\begin{proof}[Proof of Theorem~\ref{thm:caracterisation}] We first prove $(3) \Rightarrow (1) \Rightarrow (2)$. Implication $(3) \Rightarrow (1)$ follows from Lemma~\ref{lem:cycle}. Implication $(1) \Rightarrow (2)$ is trivial since $\MPs_2(\rho) \geq \MPi_2(\rho)$ for all plays $\rho$.

We thus focus on implication $(2) \Rightarrow (3)$. Suppose the existence of a play $\rho \in \Plays(G)$ and an initial credit $c_0$ such that $\forall k$, $c_0 + w_1(\rhofactor{0}{k}) \geq 0$, and $\MPs_2(\rho) > 0$. By Lemma~\ref{lem:minLocaux}, there exist infinitely many local minima for the energy. As $V$ is finite, there exists $v\in V$ and infinitely many local minima associated to $v$. We suppose that we only consider those local minima. If there exist two local minima $k, \ell$ such that $w_2(\rho[k,\ell]) > 0$, then $\rho[k,\ell]$ is a reachable good cycle by construction. Therefore, we suppose that for every pair of local minima $k,\ell$, we have $w_2(\rho[k,\ell]) \leq 0$. 
Let us denote by $k_0$ the first local minimum and let $\rho'$ be the suffix of $\rho$ starting at $k_0$, i.e. $\rho' = \rhofactor{k_0}{\infty}$. Notice that $\MPs_2(\rho') > 0$ since the mean-payoff-sup value is prefix-independent. As $\MPs_2(\rho') > 0$, for every $n \in \N$, there exists a prefix $\pi_n$ of $\rho'$, say $\pi_n = \rho'[0,i_n]$ for some $i_n \in \N$, such that $w_2(\pi_n) \geq n$. Again, since $V$ is finite, there exists a vertex $w \in V$ and infinitely many prefixes that end $w$. We suppose that we only consider those prefixes. See Figure~\ref{fig:EMP:proof1} for the construction so far.

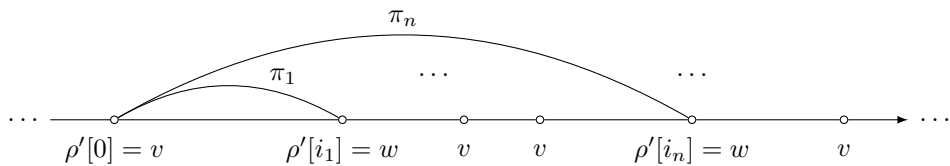
\begin{figure}[h]
\centering
\begin{tikzpicture}[scale=2]

    \draw (-1.6,0) node [circle] (A) {$\ldots$};
    \draw (-1,0) node [circle, draw,inner sep=1pt] (B) {};
    \draw (0.5,0) node [circle, draw,,inner sep=1pt] (C) {};
    \draw (1.3,0) node [circle, draw,,inner sep=1pt] (D) {};
    \draw (1.8,0) node [circle, draw,,inner sep=1pt] (E) {};
    \draw (2.8,0) node [circle, draw,,inner sep=1pt] (F) {};
    \draw (3.8,0) node [circle, draw,,inner sep=1pt] (G) {};
    \draw (4.4,0) node [circle] (H) {$\ldots$};
    \draw (2.8,0.3) node [circle] (M) {$\ldots$};
    
    
    \draw (A) -- (B);
	\draw (B) -- (C);
	\draw (C) -- (D);
	\draw (D) -- (E);
	\draw (E) -- (F);
	\draw (F) -- (G);
	\draw[->,>=latex] (G) to (H);    
	
	\draw (1.1,0.3) node [circle] (I) {$\ldots$};
	\draw (-1,-0.2) node [circle] (J) {$\rho'[0]=v$};
	\draw (1.3,-0.2) node [circle] (K) {$v$};
	\draw (1.8,-0.2) node [circle] (L) {$v$};
	\draw (3.8,-0.2) node [circle] (M) {$v$};
	\draw (0.5,-0.2) node [circle] (J) {$\rho'[i_1]=w$};
	\draw (2.8,-0.2) node [circle] (K) {$\rho'[i_n]=w$};
    
    \draw (B) to[bend left] node[above,near end] {$\pi_1$} (C);
    \draw (B) to[bend left] node[above,midway] {$\pi_n$} (F);

\end{tikzpicture}
\caption{Vertices $v$ represent local minima for the energy, and prefixes $\pi_n$ of $\rho'$ are such that $w_2(\pi_n) \geq n$ and end in vertex $w$.}
\label{fig:EMP:proof1}
\end{figure}

Since $w_1(\pi_n) \geq -c_0$ for all $n$ (as $k_0$ is a local minimum for the energy in $\rho$), there exists some $m \in \N$ such that for all  $n \in \N$,
\begin{eqnarray} \label{eqm}
w_1(\pi_m) \leq w_1(\pi_n).
\end{eqnarray}
Let $\pi_m$ be such a prefix, and consider the factor $\rho'[i_m,k]$ where $k$ is the first local minimum for the energy after $i_m$. Let us focus on the path $\rho'[0,k]$ composed of the concatenation of $\pi_m$ with $\rho'[i_m,k]$ (see Figure~\ref{fig:EMP:proof2}). Remark that $\rho'[0,k]$ is a cycle from $v$ to $v$ with 
\begin{eqnarray} \label{eqk}
w_1(\rho'[0,k]) \geq 0
\end{eqnarray}
(as $k_0$ is a local minimum for the energy). 

\begin{figure}
\centering
\begin{tikzpicture}[scale=2]

    \draw (-1.6,0) node [circle] (A) {$\ldots$};
    \draw (-1,0) node [circle, draw,inner sep=1pt] (B) {};
    \draw (2,0) node [circle, draw,,inner sep=1pt] (C) {};
    \draw (3.4,0) node [circle, draw,,inner sep=1pt] (D) {};
    \draw (4.4,0) node [circle] (E) {$\ldots$};

    
    \draw (A) -- (B);
	\draw (B) -- (C);
	\draw (C) -- (D);
	\draw[->,>=latex] (D) to (E);

	\draw (-1,-0.2) node [circle] (J) {$\rho'[0]=v$};
	\draw (3.4,-0.2) node [circle] (M) {$\rho'[k] = v$};
	\draw (1.89,-0.2) node [circle] (J) {$\rho'[i_m] \ = w$};
    
    \draw (B) to[bend left] node[below,midway] {$\pi_m$} (C);
    \draw[dashed] (B) to[bend right=80] node[above,midway] {$\pi_N$} (C);

\end{tikzpicture}
\caption{Prefix $\pi_m$ is replaced by $\pi_N$, with $w_2(\pi_N) \geq w_2(\rho'[i_m,k])$ where $k$ is the first local minimum for the energy after $i_m$.}
\label{fig:EMP:proof2}
\end{figure}
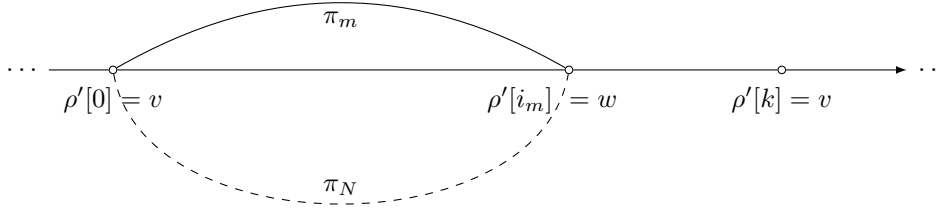

To finish the proof, it would be sufficient to show that $w_2(\rho'[0,k]) > 0$, but we assumed that such a path (from a local minimum to another one) is such that $w_2(\rho'[0,k]) \leq 0$. Thus we have  $w_2(\rho'[0,k]) = w_2(\pi_m) + w_2(\rho'[i_m,k]) \leq 0$ such that $w_2(\pi_m) \geq m$, which means that $w_2(\rho'[i_m,k]) \leq -m$. However, we are going to show that we can easily replace the prefix $\pi_m$ by another one to balance $w_2(\rho'[i_m,k])$ and obtain the desired result. Let $N \in \N$ be such that $N > |w_2(\rho'[i_m,k])|$ and consider the prefix $\pi_N$, which is such that $w_2(\pi_N) \geq N$ and $w_1(\pi_N) \geq w_1(\pi_m)$ by $(\ref{eqm})$. We construct the path $\pi^*$ as the concatenation of $\pi_N$ with $\rho'[i_m,k]$ (see Figure~\ref{fig:EMP:proof2}) and claim that this path is the desired reachable good cycle. First, $\pi^*$ is indeed a path in the graph since the last vertex of $\pi_N$ and the first vertex of $\rho'[i_m,k]$ are both equal to $w$. Second, $\pi^*$ is a cycle from $v$ to $v$ that is reachable from the initial vertex. Third, we have $w_2(\pi^*) = w_2(\pi_N) + w_2(\rho'[i_m,k]) > 0$ by construction of $N$. Finally, as $w_1(\rho'[0,k]) = w_1(\pi_m) + w_1(\rho'[i_m,k]) \geq 0$ by $(\ref{eqk})$ and $w_1(\pi_m) \leq w_1(\pi_N)$ by $(\ref{eqm})$, it yields that $w_1(\pi^*) = w_1(\pi_N) + w_1(\rho'[i_m,k]) \geq 0$.
\end{proof}

\subsection{Properties of good cycles and good multicycles}

Proof of Theorem~\ref{thm:caracterisation} does not provide any information regarding the shape of the reachable good cycle. In this section, we give such a precise description for both good cycles and good multicycles.

\begin{proposition} \label{prop:goodcycle}
Let $G$ be a graph structure. There exists a reachable good cycle if and only if
\begin{enumerate}
\item either there exists a reachable good cycle that is simple;
\item or there exist two simple cycles $C$, $C'$ that are in the same reachable strongly connected component, and such that their respective weight vectors $w(C) = (-x,y)$ and $w(C') = (x',-y')$ satisfy $x,x',y \in \Nzero$ and $y' \in \N$ and make an angle $< 180^o$;
\end{enumerate}
There exists a reachable good multicycle if and only if
\begin{enumerate}
\item either there exists a reachable good multicycle ${\cal C} = \{C\}$ composed of a unique simple cycle $C$;
\item or there exist two simple cycles $C$, $C'$ that are in the same reachable strongly connected component, and such that their respective weight vectors $w(C) = (-x,y)$ and $w(C') = (x',-y')$ satisfy $x,x', y \in \Nzero$ and $y' \in \N$ and make an angle $\leq 180^o$;
\end{enumerate}
\end{proposition}

Notice the differences between the second cases of Proposition~\ref{prop:goodcycle}: for good cycles, angle $< 180^o$, and for good multicycles, angle $\leq 180^o$. For good cycles, the second case is depicted in Figure~\ref{fig:cond2}. Let us illustrate the characterization given in Proposition~\ref{prop:goodcycle} with our two running examples.

\begin{figure}[h]
	\centering
	\begin{tikzpicture}[scale=.7]
	\node[] (a) at (-1,3){};
	\node[] (b) at (0,0){};
	\node[] (c) at (3,-2){};
	\draw [very thin,gray] (-3,-3) grid (3,3);
	\draw[->,>=latex, thick] (-3,0) to (3.4,0);
	\draw (3.4,-0.2) node [circle] (L) {$w_1$};
	\draw[->,>=latex, thick] (0,-3) to (0,3.3);
	\draw (-0.2,3.4) node [circle] (L) {$w_2$};
	\draw[->,>=latex, thick] (0,0) to node[near end,left] {$w(C)$} (-1,3);
	\draw[very thick,dotted] (0,0) to node[midway,above] {} (1,-3);
	\draw[->,>=latex, thick] (0,0) to node[near end,right] {$w(C')$} (3,-2);
	\pic [draw,angle eccentricity=1.5] {angle = c--b--a};
	\end{tikzpicture}
	\caption{Geometrical view of the second case of Proposition~\ref{prop:goodcycle} for good cycles.} 
	\label{fig:cond2}
\end{figure}
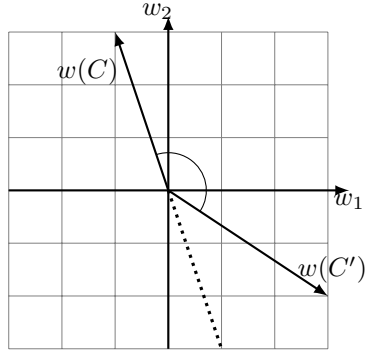

%
%
%

\begin{example}
In case of Figure~\ref{fig::EMP:finiteMemory}, the good cycle is characterized by the two cycles $C = (v_1,v_1)$, $C' = (v_0,v_0)$ with respective weights $(-1,3)$, $(1,-1)$. In case of Figure~\ref{fig::EMP:infiniteMemory}, the reachable good multicycle is characterized by the two cycles $C = (v_1,v_1)$, $C' = (v_0,v_0)$ with respective weights $(-1,1)$, $(1,-1)$. Moreover one can check that there is no reachable good cycle by inspecting the three simple cycles of the game: the two cases of Proposition~\ref{prop:goodcycle} never occur.
\end{example}

For good cycles, there is an additional result stated in the following corollary.

\begin{corollary} \label{cor:goodcycle}
Let $G$ be a player-1 game structure. There exists a reachable good cycle if and only if for the decision problem $\EMPi^{> 0}$ (resp. $\EMPs^{> 0}$), $\playerOne$ has a finite-memory winning strategy of size pseudo-polynomial in $||E||$.
\end{corollary}

Before proving Proposition~\ref{prop:goodcycle} and Corollary~\ref{cor:goodcycle}, let us establish the following lemma. 

\begin{lemma} \label{lem:goodcycle}
Let $(-x,y)$ and $(x',-y')$ be two vectors with $x,x',y \in \Nzero$ and $y' \in \N$.
\begin{itemize}
\item If these vectors make an angle $< 180^o$, then there exist $a,b \in \Nzero$ such that $a = x  x' +  y  y'$, $b = x^2 + y^2$, and $a \cdot (-x,y) + b \cdot (x',-y') > (0,0)$.
\item If these vectors make an angle $\leq 180^o$, then there exist $a,b \in \Nzero$ such that $a \cdot (-x,y) + b \cdot (x',-y') \geq (0,0)$.
\end{itemize}
\end{lemma}

\begin{proof}
We first treat the particular situation of two vectors $(-x,y)$ and $(x',-y')$ that make an angle of $180^o$. Clearly there exists $a,b \in \Nzero$ such that $a \cdot (-x,y) + b \cdot (x',-y')  = (0,0)$. In this way we get the second case of the proposition for this particular situation.

Consider now two vectors $(-x,y), (x',-y')$ that make an angle $< 180^o$. Under this hypothesis, we treat together and in the same way the two cases of the proposition. As $(-x,y), (x',-y')$ form a basis in $\Q^2$, there exist $a,b \in \Q$ such that
\begin{eqnarray} \label{eq:ab}
(y,x) = a^* \cdot (-x,y) + b^* \cdot (x',-y').
\end{eqnarray}
One can show that 
$$a^* = \frac{x  x' +  y  y'}{x'  y -  x y'}  \quad\mbox{ and }\quad b^* = \frac{x^2 + y^2}{x' y -  x y'}.$$  
Therefore by multiplying $(\ref{eq:ab})$ by $x'  y -  x y'$, we get 
\begin{eqnarray}
(x'  y -  x  y' ) \cdot (y,x)  =  a \cdot (-x,y) + b \cdot (x',-y') \label{eq:cycle}
\end{eqnarray}
with $a = x  x' +  y  y'$ and $b = x^2 + y^2$. 
Notice that $a, b$ are both in $\Nzero$ since $x,x' \in \Nzero$ and $y, y' \in \N$ by hypothesis. Moreover, $(x'  y -  x  y' ) \cdot (y,x) > (0,0)$ since $x'  y -  x  y' \in \Nzero$ due to the hypothesis that the vectors  $(-x,y)$ and $(x',-y')$ make an angle $< 180^o$.
\end{proof}

\begin{proof}[Proof of Proposition~\ref{prop:goodcycle} and Corollary~\ref{cor:goodcycle}]
We begin with the case of good cycles. The proof for good multicycles share similar arguments, we will explain the main differences later in the proof. 

\subparagraph{Case of good cycles.} 

Let $\pi$ be a reachable good cycle and let us show that cases (1) or (2) of Proposition~\ref{prop:goodcycle} occur. Let us suppose that $\pi$ is not simple (otherwise case (1) occurs), and let $\{C_1, \ldots, C_n\}$ be its cycle decomposition. As $\pi$ is a cycle, we have $w(\pi) = \sum_{k=1}^n w(C_k)$. Without lost of generality  (with respect to $w(\pi)$), we suppose that  $\{C_1, \ldots, C_n\}$ contains no cycle $C_k$ with weight $w(C_k) = (0,0)$. If there exists some cycle $C_k$ which is a good cycle, then again case $(1)$ occurs (recall that each $C_k$ is simple). Hence we assume that for all $k$, cycle $C_k$ is not a good cycle. We provide a geometrical reasoning (see Figure~\ref{fig:cond2}) to get the cycles $C, C'$ of case (2) of Proposition~\ref{prop:goodcycle}.

Let $C$ be a cycle among $\{C_1, \ldots, C_n\}$ such that $w_2(C) > 0$. Such a cycle exists, otherwise we would obtain $w_2(\pi) = \sum_{k=1}^n w_2(C_k) \leq 0$ in contradiction with $\pi$ being a good cycle. Moreover $w_1(C) < 0$, otherwise $C$ would be a good cycle. Thus $w(C) = (-x,y)$ with $x,y \in \Nzero$.  We take such a cycle $C \in \{C_1, \ldots, C_n\}$ with the minimum ratio $\frac{y}{-x}$. Now consider the line generated by $w(C)$ in $\Q^2$, with $\frac{y}{-x}$ being its slope (see Figure~\ref{fig:cond2}). 

If the weights $w(C_k)$ of all cycles $C_k$ are under or on this line, then so is $w(\pi)= \sum_{k=1}^n w(C_k)$, again in contradiction with $\pi$ being a good cycle. Therefore, let $C'$ be a cycle among $\{C_1, \ldots, C_n\}$ which is strictly above the line generated by $w(C)$. It follows that vectors $w(C), w(C')$ make an angle $< 180^o$.  Moreover $w(C')$ is in the fourth quadrant and equal to $(x',-y')$ with $x' \in \Nzero$ and $y' \in \N$ (again, see Figure~\ref{fig:cond2}). Indeed $w(C')$ is neither in the first quadrant since it is not a good cycle, nor in the second quadrant by minimality of the ratio $\frac{y}{-x}$. 

Let us now prove that if cases (1) or (2) of Proposition~\ref{prop:goodcycle} occur, then there exists a reachable good cycle. This is trivially true if case (1) occurs. Suppose that case (2) occurs. Let $C, C'$ be two cycles satisfying the conditions of case (2), and let $\pi_{C,C'}$ be a simple path from $C$ to $C'$ and $\pi_{C',C}$ be a simple path from $C'$ to $C$. We are going to show how to construct a good cycle $\pi$ from the two cycles $C, C'$ and the two paths  $\pi_{C,C'}, \pi_{C',C}$. This cycle $\pi$ is not simple and has the following shape:
\begin{enumerate}
\item Loop $\alpha$ times in $C$
\item Follow path $\pi_{C,C'}$
\item Loop $\beta$ times in $C'$
\item Follow path $\pi_{C',C}$.
\end{enumerate}
Let us explain how to choose constants $\alpha$ and $\beta$ such that $\pi$ is good and such that they are pseudo-polynomial in $||E||$. By Lemma~\ref{lem:goodcycle}, there exist $a,b \in \Nzero$ such that $a = x  x' +  y  y'$, $b = x^2 + y^2$, and 
\begin{eqnarray}
a \cdot (-x,y) + b \cdot (x',-y') = a \cdot w(C) + b \cdot w(C')  > (0,0). \label{eq:cycle}
\end{eqnarray}
Hence, (\ref{eq:cycle}) indicates that we can loop on $C$ and $C'$, i.e. combine $w(C)$ and $w(C')$, to obtain a vector of positive weights, as large as we want, in particular to balance the possibly negative weights of $\pi_{C,C'}$ and $\pi_{C',C}$. More precisely, as $\pi_{C,C'}, \pi_{C',C}$ are simple paths and thus of weight (in dimension one and two) with absolute value bounded by $(|V|-1) \cdot ||E||$, it is enough to choose constants $\alpha$ and $\beta$ as follows:
\begin{eqnarray} \label{eq:genere}
\alpha = 2 a \cdot |V| \cdot ||E|| \quad\mbox{ and }\quad \beta = 2 b \cdot  |V| \cdot ||E||.
\end{eqnarray}
In this way, we get 
\begin{eqnarray*}
w(\pi) &=& \alpha \cdot w(C) + w(\pi_{C,C'}) + \beta \cdot w (C') + w(\pi_{C',C}) \\
          &\geq&  2 |V| \cdot ||E|| \cdot (1,1) + w(\pi_{C,C'}) + w(\pi_{C',C})  \quad\quad \mbox{by (\ref{eq:cycle}) and (\ref{eq:genere})} \\
          &>& (0,0).
\end{eqnarray*} 
Therefore, $\pi$ is a reachable good cycle and constants $\alpha$ and $\beta$ are pseudo-polynomial in $||E||$ since $x, x', y, y'$ are bounded by $|V| \cdot ||E||$ (the cycles $C,C'$ are simple).

We have proved that when cases (1) or (2) of Proposition~\ref{prop:goodcycle} occur, there exists a reachable good cycle. We can go further and derive from this cycle, for both decision problems $\EMPi^{> 0}$ and $\EMPs^{> 0}$, a finite-memory winning strategy for $\playerOne$ with pseudo-polynomial size. In this way, Corollary~\ref{cor:goodcycle} will be also proved. Recall from Lemma~\ref{lem:cycle} that given a reachable good cycle, for a well-chosen initial credit, there exists a winning strategy for $\playerOne$ for problems $\EMPi^{> 0}$ and $\EMPs^{> 0}$ that consists in reaching the simple good cycle and looping in it. In case (1) of Proposition~\ref{prop:goodcycle}, this strategy is memoryless, and in case (2) of Proposition~\ref{prop:goodcycle}, it is a finite-memory strategy of size pseudo-polynomial in $||E||$ since constants $\alpha, \beta$ of (\ref{eq:genere}) are pseudo-polynomial in $||E||$.

\subparagraph{Case of good multicycles.} 
We now turn to the proof of Proposition~\ref{prop:goodcycle} for multicycles. Let ${\cal C} = \{C_1, \ldots, C_n\}$ be a reachable good multicycle (each $C_k$ is a simple cycle). Let us show that cases (1) or (2) of Proposition~\ref{prop:goodcycle} occur in a way similar to what was done previously for good cycles (the proof is thus here sketched). Without lost of generality, we suppose that  $\cal C$ contains no cycle $C_k$ with weight $w(C_k) = (0,0)$. If $\cal C$ contains a cycle $C_k$ with $w(C_k)\geq (0,0)$, then case (1) occurs. Hence we assume that for all $k$, cycle $C_k$ has weight $w(C_k) \not\geq (0,0)$. It follows that there exists $C \in {\cal C}$ such that $w_2(C) \geq 0$ and $w_1(C) < 0$. Thus $w(C) = (-x,y)$ with $x \in \Nzero$ and $y \in \N$, and we take such a cycle $C$ with the minimum ratio $\frac{y}{-x}$. We consider the line generated by $w(C)$ in $\Q^2$. 

The weights $w(C_k)$ cannot be all strictly under this line. Hence there exists some cycle $C' \in {\cal C}$ that is either strictly above the line or on it and in the direction opposite to $C$. In the first case, we conclude (as for good cycles) that $w(C') = (x',-y')$ with $x',y' \in \Nzero$ and vectors $w(C), w(C')$ make an angle $< 180^o$. Moreover, we get that $y \in \Nzero$. In the second case, the vectors $w(C), w(C')$ in opposite directions make an angle of $180^o$ and thus $w(C') = (x',-y')$ with $x' \in \Nzero$, $y' \in \N$. Moreover, as $w(C') \not\geq (0,0)$, we get $y, y' \in \Nzero$. 

Let us now prove that if cases (1) or (2) of Proposition~\ref{prop:goodcycle} occur, then there exists a reachable good multicycle. This is trivially true if case (1) occurs. Suppose that case (2) occurs. Let $C, C'$ be two cycles satisfying the conditions of case (2) and let show how to construct a good multicycle $\cal C$ from them. By Lemma~\ref{lem:goodcycle}, there exist $a,b \in \Nzero$ such that $a \cdot w(C) + b \cdot w(C')  \geq (0,0)$. Hence the required multicycle $\cal C$ is composed of $a$ occurrences of cycle $C$ and $b$ occurrences of cycle $C'$.
\end{proof}

\subsection{Characterization in case of non-strict inequalities}

We now prove the characterization given in Theorem~\ref{thm:caracterisationLarge} in case of a mean-payoff objectives with non-strict inequality.

\begin{proof}[Proof of Theorem~\ref{thm:caracterisationLarge}]
We prove that $(1) \Rightarrow (2) \Rightarrow (3) \Rightarrow (1)$. Implication $(1) \Rightarrow (2)$ is immediate since $\MPs_2(\rho) \geq \MPi_2(\rho)$ for all plays $\rho$. 

Let us prove implication $(2) \Rightarrow (3)$. Suppose that $\playerOne$ is winning for the objective $\energy_1(c_0) \cap \MPs_2(\geq 0)$ for some $c_0$. We are going to show that there exists a reachable good multicycle as described in Proposition~\ref{prop:goodcycle}. From hypothesis (2), it follows that for all $\epsilon \in \Q$, $\epsilon > 0$, $\playerOne$ is winning for the objective $\energy_1(c_0) \cap \MPs_2(> -\epsilon)$ (using threshold $\neq 0$ is allowed by Remark~\ref{rem:threshold}). We consider the game structure $G_\epsilon$ obtained from $G$ by replacing function $w_2$ by function $w_2^\epsilon = w_2 + \epsilon$ (function $w_1$ is left unchanged; using rational weights are allowed by Remark~\ref{rem:threshold}). Hence $\playerOne$ is winning in this game $G_\epsilon$ for the objective $\energy_1(c_0) \cap \MPs_2(> 0)$, and by Theorem~\ref{thm:caracterisation} there exists a reachable good cycle in $G_\epsilon$ as described in Proposition~\ref{prop:goodcycle}.

We first assume that there exists $\epsilon \leq \frac{1}{|V|}$ such that this reachable good cycle $C$ is simple, thus with length $k \leq |V|$. Let us show that ${\cal C} = \{C\}$ is the required multicycle in $G$ with $w({\cal C}) \geq (0,0)$. As $C$ is a good cycle in $G_\epsilon$,  we have $w_1(C) \geq 0$ and $w_2^\epsilon(C) > 0$. Hence
$$w_2^\epsilon(C)  = w_2(C) + k \cdot \epsilon > 0.$$ 
Suppose that in $G$, $w_2(C) < 0$, i.e. $w_2(C) \leq -1$. Then in $G_\epsilon$, $w_2^\epsilon(C) \leq -1 + k \cdot \epsilon \leq -1 + |V| \cdot  \frac{1}{|V|} = 0$ which is impossible. It follows that $w(C) \geq (0,0)$ as announced.

We then assume that for all $\epsilon \leq \frac{1}{|V|}$, there is no reachable good cycle $C_\epsilon$ that is simple. Then we have two simple cycles $C_\epsilon, C'_\epsilon$ as described in Proposition~\ref{prop:goodcycle}. As there is a finite number of simple cycles in $G$, there exists a sequence $\epsilon_n \rightarrow 0$ using the same pair of cycles $C,C'$ with respective length $k,k' \leq |V|$ such that in game $G_{\epsilon_n}$, 
\begin{eqnarray} \label{eq:epsilon}
w^{\epsilon_n}(C) = (-x,y + k \cdot \epsilon_n) \quad \mbox{ and } \quad  w^{\epsilon_n}(C') = (x',-y' + k' \cdot \epsilon_n)
\end{eqnarray} 
such that $x,x' \in \Nzero$, $y + k \cdot \epsilon_n > 0$, $y' - k' \cdot \epsilon_n \geq 0$, and vectors $w^{\epsilon_n}(C), w^{\epsilon_n}(C')$ make an angle $< 180^o$. The latter condition is equivalent to
\begin{eqnarray} 
x' \cdot (y + k \cdot \epsilon_n) - x \cdot (y' - k' \cdot \epsilon_n) &>& 0 \nonumber \\
x' y - x y' + (x'  k + x  k') \cdot \epsilon_n &>& 0. \label{eq:angle}
\end{eqnarray}
When $\epsilon_n \rightarrow 0$ in (\ref{eq:epsilon}) and (\ref{eq:angle}), we get in the game $G$ that $w(C) = (-x,y)$, $w(C') = (x',-y')$ with $x,x'\in \Nzero$ and $y, y' \in \N$, and $x' y - x y' \geq 0$ showing that vectors $w(C), w(C')$ make an angle $\leq 180^o$. If $y = 0$, it follows that $y' = 0$ by the angle $\leq 180^o$, and thus ${\cal C} = \{C'\}$ is a reachable good multicycle. If $y \in \Nzero$, it follows from case (2) of Proposition~\ref{prop:goodcycle} that there exists a good reachable multicycle in $G$.

We now prove implication $(3) \Rightarrow (1)$. Suppose that there exists a reachable good multicycle and let us construct a winning strategy for $\playerOne$ for the decision problem $\EMPi^{\geq 0}$. We apply Proposition~\ref{prop:goodcycle}. If the reachable good multicycle is composed of a unique simple cycle $C$ with weight $w(C) \geq (0,0)$, then $\playerOne$ has a memoryless winning strategy (by Lemma~\ref{lem:cycle}). So let us suppose by Proposition~\ref{prop:goodcycle} that there exist two simple cycles $C$, $C'$ in the same reachable connected component such that their weight vectors $w(C) = (-x,y)$ and $w(C') = (x',-y')$ satisfy $x,x',y \in \Nzero$ and $y' \in \N$ and make an angle $\leq 180^o$. Moreover by Lemma~\ref{lem:goodcycle}, there exist $\alpha, \beta \in \Nzero$ such that
\begin{eqnarray}
\alpha \cdot w(C) + \beta \cdot w(C') \geq (0,0). \label{eq:positive} 
\end{eqnarray}
As  $C, C'$ are in the same reachable connected component, let $\pi_{C,C'}$ be a simple path from $C$ to $C'$, $\pi_{C',C}$ be a simple path from $C'$ to $C$, and $\pi_{0,C'}$ be a simple path from the initial vertex $v_0$ to $C'$. All those paths have their weight bounded by $|V| \cdot ||E||$. To balance the possibly negative energy of $\pi_{C,C'}$ and $\pi_{C',C}$, we choose $\gamma \in \N$ such that
\begin{eqnarray}
\gamma \cdot w_1(C') = \gamma \cdot x' \geq (|\pi_{C,C'}|+|\pi_{C',C}|) \cdot ||E||. \label{eq:gamma} 
\end{eqnarray}
Notice that $\gamma \leq 2|V|\cdot ||E||$. Consider the following strategy $\sigma_1$ for $\playerOne$: 
\begin{enumerate}
\item Follow path $\pi_{0,C'}$
\item Initialize $Z = 1$
\item At round $Z$
\begin{enumerate}
\item Loop $Z \cdot \beta + \gamma$ times in cycle $C'$
\item Follow path  $\pi_{C',C}$
\item Loop $Z \cdot \alpha$ times in cycle $C$
\item Follow path  $\pi_{C,C'}$
\end{enumerate}
\item Increment $Z$ by $1$ and goto 3.
\end{enumerate}
Notice that Example~\ref{ex:meminfinie} is a particular case of the studied situation with $\alpha = \beta = 1$ and $\gamma = 0$. With arguments similar to the ones done for Example~\ref{ex:meminfinie}, let us prove that the play $\rho$ consistent with $\sigma_1$  belongs to $\Obj = \energy_1(c_0) \cap \MPi_2 (\geq 0)$ with $c_0 =  |w_1(\pi_{0,C'})|$. By definition of $c_0$ and thanks to (\ref{eq:positive}) and (\ref{eq:gamma}), the energy level on the first dimension never drops below zero. Thus we only focus on the second dimension. Consider any prefix $\pi = \rhofactor{0}{k}$ of $\rho$ during round $Z$. The length $k$ is $\pi$ is upper bounded by
\begin{eqnarray}
k            \leq |\pi_{0,C'}| + \sum_{\ell=1}^{Z} ((\ell  \beta + \gamma)\cdot |C'| + |\pi_{C',C}| + \ell  \alpha \cdot |C| + |\pi_{C,C'}| ) \in {\cal O}(Z^2).\label{eq:denominateur}
\end{eqnarray}
The weight $w_2(\pi)$ is lower bounded by   
\begin{eqnarray}
w_2(\pi)    \geq  - |V|\cdot ||E|| ~-~ (Z-1) \cdot (\gamma + 2) \cdot |V|\cdot ||E|| ~-~ (Z\beta + \gamma  + 1)  \cdot |V|\cdot ||E|| ~\in {\cal O}(Z).\label{eq:numerateur}
\end{eqnarray}
Indeed, the first term is a lower bound for $w_2(\pi_{0,C'})$, in view of (\ref{eq:positive}) the second term is a lower bound on the remaining negative weight after rounds $1, 2, \ldots, Z-1$, and the last term is a lower bound on the worst weight during round $Z$ (just after (b)). By (\ref{eq:denominateur}) and (\ref{eq:numerateur}), it follows that the mean-payoff-inf value  $\MPi_2(\rho)$ of $\rho$ is $\geq 0$ since the average weight $\frac{w_2(\pi)}{k}$ during round $Z$ is lower bounded by a quantity of the form $\frac{-Z}{Z^2}$ which converges to $0$. This shows that $\rho \in \Obj$.
\end{proof}

\subsection{Proof of Theorem~\ref{thm:onePlayer}}

We now have all the ingredients to prove Theorem~\ref{thm:onePlayer} which is the main result of Section~\ref{sec:OneP}. By Theorems~\ref{thm:caracterisation} and~\ref{thm:caracterisationLarge}, solving the energy mean-payoff desicion problem for player-1 game structures reduces to decide whether there exists a reachable good cycle or multicycle. We will show that this can be tested in polynomial time thanks to a result in~\cite{KosarajuS88}. In case of mean-payoff objectives with strict inequality, when a reachable good cycle exists, we know that $\playerOne$ has a winning strategy with memory size pseudo-polynomial in $||E||$ by Corollary~\ref{cor:goodcycle}, and we will provide an example of game where pseudo-polynomial memory is necessary for $\playerOne$ to win. In case of mean-payoff objectives with non-strict inequality, when a reachable good multicycle exists, we know that $\playerOne$ may need a strategy with infinite memory to win by Example~\ref{ex:meminfinie}.

\begin{theorem}{\cite{KosarajuS88}} \label{thm:cyclepoly}
Let $G$ be a game structure.
\begin{enumerate}
\item Deciding whether $G$ contains a multicycle $\pi$ with $w(\pi) = (0,0)$  can be done in polynomial time.
\item Deciding whether $G$ contains a cycle $\pi$ with $w(\pi) = (0,0)$  can be done in polynomial time.
\end{enumerate}
\end{theorem}

\begin{proof}[Proof of Theorem~\ref{thm:onePlayer}] We begin by showing that the energy mean-payoff decision problem can be solved in polynomial time in case of strict inequalities.
By Theorem~\ref{thm:caracterisation}, solving the decision problem $\EMPi^{>0}$ or $\EMPs^{>0}$ is equivalent to testing the existence a reachable good cycle in $G$. The latter property can be checked in polynomial time as follows. Let $G'$ be the graph composed of the vertices of $G$ reachable from the initial state $v_0$. We derive two graphs $G'_1$ and $G'_2$ from $G'$ where
\begin{itemize}
\item in $G'_1$, one self-loop with weight $(-1,0)$ is added to each vertex of $G'$;
\item in $G'_2$, two self-loops with respective weights $(-1,0)$ and $(0,-1)$ are added to each vertex of $G'$.
\end{itemize}
Clearly, $G'_1$ and $G'_2$ can be computed in polynomial time. One can easily verify that there is a good cycle in $G'$, i.e. a cycle $\pi$ with $w_1(\pi) \geq 0$ and $w_2(\pi) > 0$, if and only if there is a cycle $\pi_2$ in $G'_2$ with $w(\pi_2) = (0,0)$, but no cycle $\pi_1$ in $G'_1$ with $w(\pi_1) = (0,0)$. Indeed, the first condition guarantees the existence of a cycle $\pi$ in $G'$ with weight $w(\pi) \geq (0,0)$ whereas the second condition guarantees that $w_2(\pi) \neq 0$. These tests can be done in polynomial time by the second statement of Theorem~\ref{thm:cyclepoly}.

The arguments are simpler for proving that the energy mean-payoff decision problem can be solved in polynomial time in case of non-strict inequalities.
By Theorem~\ref{thm:caracterisationLarge}, solving the decision problem $\EMPi^{\geq 0}$ or $\EMPs^{\geq 0}$ is equivalent to testing the existence of a reachable good multicyle.  Let $G'$ be the graph composed of the vertices of $G$ reachable from the initial state $v_0$ and such that self-loops with respective weights $(-1,0)$ and $(0,-1)$ are added to each vertex. This graph can be computed in polynomial time, and  there is a multicycle in $G'$ with weight $(0,0)$ if and only there exists a reachable good multicycle in $G$. This test can be done in polynomial time by the first statement of Theorem~\ref{thm:cyclepoly}.

We now turn to the memory requirements of winning strategies for $\playerOne$. In case of non-strict inequalities for the mean-payoff objective, Example~\ref{ex:meminfinie} indicates that infinite memory is necessary for $\playerOne$ to win. In case of strict inequalities, finite-memory strategies with size pseudo-polynomial in $||E||$ are sufficient for $\playerOne$ to win by Corollary~\ref{cor:goodcycle}.

\begin{figure}[h]
\centering
  \begin{tikzpicture}[scale=4]
    \everymath{\scriptstyle}
    \draw (0,0) node [circle, draw] (A) {$v_0$};
    \draw (0.75,0) node [circle, draw] (B) {$v_1$};
    
	\draw[->,>=latex] (A) to[bend left] node[above,midway] {$(W,-W)$} (B);
	\draw[->,>=latex] (B) to[bend left] node[below,midway] {$(W,-W)$} (A);
	
    \draw[->,>=latex] (B) .. controls +(45:0.4cm) and +(135:0.4cm) .. (B) node[above,midway] {$(-1,1)$};
	\path (-0.2,0) edge [->,>=latex] (A);    
    
    \end{tikzpicture}
\caption{Player-1 game structure where $\playerOne$ needs pseudo-polynomial memory to win.}
\label{fig::EMP:pseudoPolyMemory}
\end{figure}
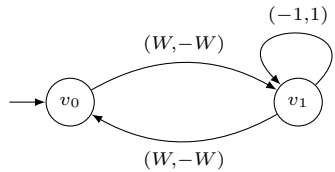

It remains to prove that pseudo-polynomial memory is necessary in case of strict inequalities. Consider the player-1 game structure on Figure~\ref{fig::EMP:pseudoPolyMemory} with $||E|| = W \in \Nzero$. In this game, $\playerOne$ has a finite-memory winning strategy for the objective $\Obj = \energy_1(c_0) \cap \MPi(> -\epsilon)$ for $c_0 = W$ and for all $\epsilon \in \Q, \epsilon > 0$ (using threshold $\neq 0$ is allowed by Remark~\ref{rem:threshold}). Indeed his winning strategy consists in repeating the following cycle $\pi$: go from $v_0$ to $v_1$, loop $2 W$ times in $v_1$, and go back to $v_0$. As this cycle $\pi$ has weight $w(\pi) = (0,0)$, the energy objective is satisfied with the initial credit $c_0 = W$, and the mean-payoff-inf objective is satisfied with non-strict inequality $\geq 0$, and thus with strict inequality $> -\epsilon$ for all $\epsilon$. One can check that this strategy uses a Moore machine with $2||E|| + 1$ memory states.

Let us prove that $\playerOne$ has no finite-memory strategy $\sigma_1$ with size $||E||$ to positively solve the problem $\EMPi^{>-\epsilon}$. Assume the contrary and take $\epsilon = \frac{1}{2W}$. Suppose that the cycle $\pi$ infinitely repeated by strategy $\sigma_1$ has a cycle decomposition using $\alpha$ occurrences of cycle $v_0v_1v_0$ and $\beta$ occurrences of cycle $v_1v_1$, that is, 
\begin{eqnarray}
w(\pi) = \alpha \cdot (2W, -2W) + \beta \cdot (-1,1) \label{eq:poids}
\end{eqnarray}
As this cycle is simple in the graph $G(\sigma_1)$ (equal to the product of $G$ with the Moore machine of $\sigma_1$), we have 
\begin{eqnarray}
0 <  |\pi| =  2\alpha + \beta \leq |V| \cdot ||E|| = 2 W \label{eq:simple}
\end{eqnarray}
As $\sigma_1$ is winning, we have by $(\ref{eq:poids})$
\begin{eqnarray}
\alpha \cdot 2W - \beta &\geq& 0 \label{eq:energy}\\
\frac{-\alpha \cdot 2W + \beta}{2 \alpha + \beta} &>& - \epsilon \label{eq:MP}
\end{eqnarray}
If $\alpha  = 0$, then in (\ref{eq:simple}) $\beta > 0$ and in (\ref{eq:energy}) $\beta \leq 0$, which is impossible. Hence $\alpha \geq 1$. If $\alpha \cdot 2W - \beta = 0$, then $\beta = \alpha \cdot 2W \geq 2W$ in contradiction with $\alpha \geq 1$ and $2\alpha  + \beta \leq  2 W$ in (\ref{eq:simple}). Hence $\alpha \cdot 2W - \beta \geq 1$. It follows with (\ref{eq:MP})  that
$$\frac{1}{2 \alpha + \beta} \leq \frac{\alpha \cdot 2W - \beta}{2 \alpha + \beta} < \epsilon = \frac{1}{2W}.$$
Therefore $2W < 2 \alpha + \beta$ in contradiction with (\ref{eq:simple}).

\end{proof}



\section{Two-player setting}\label{sec:twoP}

In this section we consider two-player energy mean-payoff games. We show that the four variants of the energy mean-payoff decision problem are in co-NP. To establish this, we show that if the answer to this problem is {\sf No}, then $\playerTwo$ has a spoiling memoryless strategy $\sigma_2$ that he can use for all initial credits $c_0 \in \N$. In the game structure $G(\sigma_2)$, $\playerOne$ is then the only player and we can apply the results of the previous section, in particular Theorem~\ref{thm:onePlayer}. 
We also show that in case of mean-payoff objectives with strict inequality, the energy mean-payoff decision problem can be reduced to the unknown initial credit problem for 4-dimensional energy games. If follows by~\cite{JurdzinskiLS15} that our decision problem can be solved in pseudo-polynomial time and that finite-memory winning strategies with pseudo-polynomial size for $\playerOne$ exist and can effectively be constructed. In case of mean-payoff objectives with non-strict inequality, we already know that infinite memory is necessary for $\playerOne$ in player-1 energy mean-payoff games by Theorem~\ref{thm:onePlayer}. We show how to construct such strategies. The results that we establish in this section are summarized in the following theorem.

\begin{theorem} \label{thm:twoPlayer}
The energy mean-payoff decision problem for two-player game structures is in co-NP. Moreover,
\begin{itemize}
\item both problems  $\EMPi^{> 0}$ and $\EMPs^{> 0}$ can be solved in pseudo-polynomial time and exponential-memory strategies are sufficient for $\playerOne$ to win;
\item for both problems $\EMPi^{\geq 0}$ and $\EMPs^{\geq 0}$, in general, $\playerOne$ needs infinite memory to win.
\end{itemize}
\noindent
In all cases, winning strategies can be effectively constructed for both players.
\end{theorem}

The proof of this result is detailed in the following sections. 

\subsection{Memoryless winning strategies for $\playerTwo$}

For all four variants of mean-payoff energy objective, we here establish that $\playerTwo$ does not need any memory for his winning strategies. Therefore, thanks to Theorem~\ref{thm:onePlayer}, the energy mean-payoff decision problem can be solved in co-NP.

%

\begin{proposition}\label{prop:memoryless}
Let ${\sim} \in \{>,\geq\}$. For all energy mean-payoff games $G$ and all initial vertices $v_0$, if the answer to the energy mean-payoff problem $\EMPi^{\sim 0}$ (resp. $\EMPs^{\sim 0}$) is {\sf No}, then there exists a memoryless strategy $\sigma_2$ for $\playerTwo$ such that for all initial credits $c_0 \in \N$, no play $\rho$ consistent with $\sigma_2$ from $v_0$ belongs to $\Obj = \energy_1(c_0) \cap \MPi_2(\sim 0)$ (resp. to $\Obj = \energy_1(c_0) \cap \MPs_2(\sim 0)$).
\end{proposition}

As a preambule to the proof of this proposition, we state the following lemma. 

\begin{lemma}
\label{lem:variants-memoryless}
For all energy mean-payoff games $G$ and initial vertices $v_0$,
let ${\sim} \in \{>, \geq \}$ and let $\sigma_2$ be a memoryless strategy for $\playerTwo$. Then 
$\sigma_2$ is winning from $v_0$ for $\energy_1(c_0) \cap \MPi_2(\sim 0)$ for all initial credits $c_0$
if and only if $\sigma_2$ is winning from $v_0$ for $\energy_1(c_0) \cap \MPs_2(\sim 0)$ for all initial credits $c_0$.
\end{lemma}

The proof of this lemma is immediate: consider the game structure $G(\sigma_2)$ induced by a memoryless strategy $\sigma_2$ for $\playerTwo$ and apply Theorem~\ref{thm:caracterisation} and Theorem~\ref{thm:caracterisationLarge} in $G(\sigma_2)$.

We now proceed to the proof of Proposition~\ref{prop:memoryless}. Note that energy objectives are not prefix-independent objectives and this proposition does not directly follow from the results of~\cite{KopczynskiN14}. However our proof is an adaptation of the proof technique of~\cite{Brazdil2010,ChatterjeeRR14,GimbertZ05,KopczynskiN14}.

\begin{proof}[Proof of Proposition~\ref{prop:memoryless}]
We only need to establish the result for the problem $\EMPi^{\sim 0}$ as we can then directly obtain the result for the problem $\EMPs^{\sim 0}$ using Lemma~\ref{lem:variants-memoryless}. Let us denote by $\Obj(c_0)$ the objective $\energy_1(c_0) \cap \MPi_2(\sim 0)$.

We prove the proposition by induction on the number $\kappa = |E| - |V|$. Suppose that the answer to the decision problem is {\sf No}. 

If $\kappa = 0$, then every vertex belonging to $\playerTwo$ has a unique outgoing edge, and $\playerTwo$ has only one (memoryless) strategy $\sigma_2$. Therefore for all initial credits $c_0 \in \N$, no play $\rho$ consistent with $\sigma_2$ from $v_0$ belongs to $\Obj(c_0)$.

Suppose now the statement of Proposition~\ref{prop:memoryless} holds for $\kappa \leq n$ for some $n \in \N$ and let us prove that it is true for $\kappa = n+1$. For this purpose, let $G$ be a game structure such that $|E| - |V| = n+1$ and let $v_0$ be the initial vertex. If every vertex $v \in V_2$ has a unique outgoing edge, we are done as before. So suppose that there exists some vertex $v^* \in V_2$ that has at least two outgoing edges. We partition this set of edges into two non-empty subsets $E_\ell$ and $E_r$ and we define from $G$ two smaller game structures, denoted $G_\ell$ and $G_r$, with the same vertices and edges except that the set of outgoing edges from $v^*$ is restricted to  $E_\ell$ and $E_r$  respectively. By construction of $G_\ell$ and $G_r$, we have $|E_j| - |V_j| \leq n$ for $j \in \{\ell,r\}$, and so we can use the induction hypothesis on them. 

Suppose first that the answer to the decision problem is also {\sf No} in either $G_\ell$ or $G_r$. Then by induction hypothesis $\playerTwo$ has a memoryless winning strategy $\sigma_2$ in $G_\ell$ (resp. in $G_r$). As $\sigma_2$ is also winning for him in the whole game $G$, we are done.  

Suppose now that the answer to the decision problem is {\sf Yes} in both $G_\ell$ and $G_r$. Hence for each $j \in \{\ell,r\}$, let $\sigma_1^j$ be a winning strategy for $\playerOne$ from $v_0$ in $G_j$ for the objective $\Obj(c_0^j)$ for some $c_0^j \in \mathbb{N}$. We will show that $\playerOne$ is also winning in $G$ for the objective $\Obj(d_0)$ for some well-chosen $d_0$. This is in contradiction with the negative answer to the decision problem in $G$, hence only the previous situation holds and we are done. 

If for some $j \in \{\ell,r\}$, each play from $v_0$ in $G_j$ consistent with $\sigma_1^j$ never visits $v^*$, then $\sigma_1^j$ is also winning for $\playerOne$ in the whole game $G$ and we are done. So suppose that this is not the case: it follows that $\playerOne$ is also winning from $v^*$ in each $G_j$, $j \in \{\ell,r\}$. We denote by $\tau_1^j$ a winning strategy for him from $v^*$ in $G_j$ for minimal initial credit $c_j^*$. We can assume without lost of generality that $c_\ell^* \geq c_r^*$, that is  $c_\ell^* = c_r^* + \Delta$ for some $\Delta \geq 0$. 

Let us show that $\playerOne$ is winning in $G$ for the objective $\Obj(d_0)$ with $d_0 = c_0^\ell$. For this purpose, let us fix some notations. Let $\pi v$ be a path in $V^* \cdot V_1$ that visits $v^*$. We decompose $\pi v$ as a prefix $\pi_{v_0,v^*}$ from $v_0$ to the first visit of $v^*$, a (possibly empty) sequence $C_1,\ldots,C_n$ of cycles from one visit of $v^*$ to the next one, and a suffix $\pi_{v^*,v}$ from the last visit of $v^*$ to the last vertex $v$ of $\pi v$. We label each of the paths $C_1,\ldots,C_n,\pi_{v^*,v}$ with $\ell$ (resp. $r$) if its first edge belongs to $E_\ell$ (resp. $E_r$). We denote by $\pi'_\ell$ (resp. $\pi'_r$) the path constructed from $\pi$ by removing $\pi_{v_0,v^*}$ and all the paths $C_1,\ldots,C_n,\pi_{v^*,v}$ that are labeled by $r$ (resp. $\ell$). In this way $\pi'_j$ is a path in $G_j$ from $v^*$, for both $j \in \{\ell,r\}$. A similar decomposition can be done for a play $\rho$ that visits $v^*$. Two cases occur: either there is an infinite number of cycles $C_1,C_2, \ldots$, or there is a finite number of cycles $C_1,\ldots,C_n$ followed by a suffix $\rho_{v^*}$ of $\rho$ from the last visit of $v^*$. As done with the path $\pi v$, we denote by $\rho'_\ell$ (resp. $\rho'_r$) the play (or path) constructed from $\rho$ by removing $\rho_{v_0,v^*}$ as well as all $C_i$ and $\rho_v^*$ (if it exists) labeled by $r$ (resp. $l$). 

We can now construct a winning strategy $\lambda_1$ of $\playerOne$ from $v_0$ in $G$ for the objective $\Obj(c_0^\ell)$ as follows. Let $\pi v$ be a path in $V^* \cdot V_1$. 
\begin{itemize}
\item If $\pi v$ does not visit $v^*$, we let $\lambda_1(\pi v) = \sigma_1^\ell(\pi).$
\item If $\pi v$ visits $v^*$, consider its decomposition into $\pi_{v_0,v^*}$, $C_1,\ldots,C_n$, and $\pi_{v^*,v}$. If $j \in \{r,l\}$ is the label of $\pi_{v^*,v}$, we let $\lambda_1(\pi v) = \tau_1^j(\pi'_j)$.
\end{itemize}
Let $\rho$ be a play from $v_0$ in $G$ consistent with $\lambda_1$. Let us show that $\rho$ is winning for $\Obj(c_0^\ell)$. If $\rho$ does not visit $v^*$, then by definition of $\lambda_1$, $\rho$ is a play in $G_\ell$ consistent $\sigma_1^\ell$ from $v_0$ and is thus winning for $\Obj(c_0^\ell)$. If $\pi$ visits $v^*$, we decompose $\rho$ as explained previously as a prefix $\rho_{v_0,v^*}$, followed by a finite or infinite sequence of cycles $C_i$, and an eventual  suffix $\rho_{v^*}$. We also consider $\rho'_\ell$ and $\rho'_r$. 

First notice that the energy objective $\energy_1(c_0^\ell)$ is satisfied, that is,  the energy level remains always nonnegative along $\rho$. Indeed by definition of $\lambda_1$, $(i)$ this is the case along $\rho_{v_0,v^*}$ since $\sigma_1^\ell$ is winning from $v_0$ in $G_\ell$ for the objective $\Obj(c_0^\ell)$, furthermore  the energy level at the end of $\rho_{v_0,v^*}$ is $\geq c^*_\ell = c^*_r + \Delta$ by Lemma~\ref{lem:energy}, $(ii)$ the same conclusions hold for each $C_i$ and for $\rho_{v^*}$ by Lemma~\ref{lem:energy} and since $\tau_1^j$ is winning from $v^*$ in $G_j$ for the objective $\Obj(c^*_j)$ for both $j \in \{\ell,r\}$.

Second the mean-payoff objective $\MPi_2(\sim 0)$ is also satisfied. Suppose that either $\rho'_\ell$ or $\rho'_r$ is finite, that is, the decomposition of $\rho$ ends with the suffix $\rho_{v^*}$. Then by definition of $\lambda_1$, if $\rho_{v^*}$ is labeled by $j$, then it is a suffix of $\rho'_j$ that is consistent with the winning strategy $\tau_1^j$ from $v^*$ in $G_j$. As the mean-payoff objective is prefix-independent and $\rho'_j$ belongs to $\MPi_2(\sim 0)$, $\rho$ also belongs to $\MPi_2(\sim 0)$. 
Suppose now that  both $\rho'_\ell$ or $\rho'_r$ are infinite. By definition of $\lambda_1$, each $\rho'_j$, $j \in \{\ell,r\}$, is consistent with the winning strategy $\tau_1^j$ in $G_j$, and thus belongs to $\MPi_2(\sim 0)$. Moreover, as the mean-payoff-inf objective is prefix-independent and convex, we have that $\rho$ also belongs to $\MPi_2(\sim 0)$ (see e.g.~\cite{VelnerC0HRR15} for a proof). 
We recall that an objective $\Obj$ is convex if for all plays $\rho_{\it odd}=\pi_1 \pi_3 \dots  \in \Obj$ and $\rho_{\it even}= \pi_0 \pi_2 \dots \in \Obj$, then we have $\rho=\pi_0 \pi_1 \pi_2 \pi_3 \dots \in \Obj$.
\end{proof}

Notice that from Proposition~\ref{prop:memoryless} and Lemma~\ref{lem:variants-memoryless} we directly get the following corollary.

\begin{corollary} \label{cor:variants}
For all energy mean-payoff games $G$ and initial vertices $v_0$, let ${\sim} \in \{>, \geq \}$. Then 
$\playerOne$ is winning from $v_0$ for $\energy_1(c_0) \cap \MPi_2(\sim 0)$ for some initial credit $c_0$
if and only if he is winning from $v_0$ for $\energy_1(c_0) \cap \MPs_2(\sim 0)$ for some initial credit $c_0$.
\end{corollary}

While Proposition~\ref{prop:memoryless} allows us to obtain the membership in co-NP of the decision problems and to effectively construct winning memoryless strategies for $\playerTwo$, unfortunately it does not tell us how $\playerOne$ must play from a winning vertex (when spoiling strategies do not exist for $\playerTwo$).
In the following two sections we provide results that show how $\playerOne$ needs to play in order to win energy mean-payoff games. 
We first show that $\playerOne$ can win with finite memory for the case of strict inequalities, and then we provide infinite-memory winning strategies for the case of non-strict inequalities. For the later case, we already know that infinite memory is necessary even player-1 game structures (see Theorem~\ref{thm:onePlayer}).

\subsection{Strategies for $\playerOne$: case of strict inequalities}

In case of strict inequalities, our solution is based on a reduction to multi-dimensional energy games~\cite{ChatterjeeDHR10} for which we know how to construct strategies for $\playerOne$.

\subsubsection{Multi-dimensional energy games}

We need to recall the concept of \emph{$d$-dimensional energy games}, with $d \in \Nzero$. Those games are played on  $d$-dimensional game structure $G = (V, V_1,V_2, E, w)$ where the weight function $w : E \rightarrow \Z^d$ assigns a $d$-tuple (instead of a pair) of weights $w(e)$ to each edge $e \in E$. The \emph{unknown initial credit problem} asks, given a $d$-dimensional game structure and an initial vertex $v_0$, to decide whether there exists an initial credit $c_0 = (c_{0,1},\ldots,c_{0,d}) \in \N^d$ and a winning strategy for $\playerOne$ for the objective $\Obj = \cap_{j=1}^d \energy_j(c_{0,j})$. When $d =1$ and the answer to this problem is {\sf Yes}, we denote by $c(v_0) \in \N$ the minimum initial credit for which $\playerOne$ has a winning strategy from $v_0$.
 The complexity of this problem has been first studied in~\cite{ChatterjeeDHR10,ChatterjeeRR14,VelnerC0HRR15} and then in~\cite{JurdzinskiLS15} for a fixed number of dimensions. 

\begin{theorem}[\cite{ChatterjeeDHR10,ChatterjeeRR14,JurdzinskiLS15,VelnerC0HRR15}] \label{thm:multiEnergy}
The unknown initial credit problem for $d$-dimensional energy games can be solved in pseudo-polynomial time, that is in time $(|V|\cdot ||E||)^{{\cal O}(d^4)}$. If the answer to this problem is 
\begin{itemize}
\item {\sf Yes}, then exponential-memory strategies are sufficient and necessary for player~$\playerOne$ to win,
\item {\sf No}, then $\playerTwo$ has a spoiling memoryless strategy $\sigma_2$ that he can use for all initial credits $c_0 \in \N^d.$
\end{itemize}
\end{theorem}  

We recall two useful lemmas.

\begin{lemma}[\cite{VelnerC0HRR15}]\label{lem:energyCycle}
Let $G$ be a player-1 $d$-dimensional energy game. Then the answer to the unknown initial credit problem is {\sf Yes} if and only if there exists a reachable cycle $\pi$ in $G$ such that $w(\pi) \geq (0,\ldots,0)$.
\end{lemma}

\begin{lemma}[\cite{ChatterjeeD12}]\label{lem:energy}
Let $G$ be a $1$-dimensional energy game and $v_0$ be an initial vertex. For all plays $\rho$ consistent with a winning strategy $\sigma_1$ for $\playerOne$, if the initial credit is $c(v_0) + \Delta$ for $\Delta \geq 0$, then the energy level at all positions of $\rho$ where a state $v$ occurs is at least $c(v) + \Delta$.
\end{lemma}

The next proposition shows that we can reduce energy mean-payoff games with strict inequality constraints to energy games with 4 dimensions.

\begin{proposition} \label{prop:reductionEnergy}
The problems $\EMPi^{>0}$ and $\EMPs^{>0}$ for energy mean-payoff games are both polynomially reducible to the unknown initial credit problem for $4$-dimensional energy games. Moreover, for the energy game $G'$ constructed from the given $G$, we have  $||E'|| = ||E||$ and $|V'|, |E'|$ are linear in $|V|, |E|$, and from a finite-memory winning strategy $\sigma'_1$ of $\playerOne$ in $G'$, we can derive a finite-memory winning strategy $\sigma_1$ of $\playerOne$ in $G$ such that the memory size of $\sigma_1$ is upper bounded by the memory size of $\sigma'_1$.
%
%
\end{proposition}

\begin{proof}
We first explain the reduction. Given an energy mean-payoff game structure $G = (V, V_1,V_2, E, w)$ with $w : E \rightarrow \Z^2$, we construct a $4$-dimensional energy game $G' = (V', V'_1,V'_2, E', w')$ with $w' : E' \rightarrow \Z^4$ as follows. Each edge $e = (v,v') \in E$ labeled by $w(e) = (x,y)$ is replaced by:
\begin{itemize}
\item five edges $(v,r), (r,s), (s,s), (s,r)$, and $(r,v')$ where $r,s$ are two new vertices,
\item such that  $w'(v,r) = (x,y,-1,1)$,  $w'(r,s) = (0,-1,0,0)$, $w'(s,s) = (0,0,1,-1)$, $w'(s,r) = (0,0,0,0)$, and $w'(r,v') = (0,0,0,0)$.
\end{itemize}
\noindent
This is illustrated in Fig.~\ref{fig:gadget}.

The set $V'_2$ is equal to $V_2$, and $V'_1$ is composed of all vertices of $V_1$ and the $2 \cdot |E|$ new vertices (two for each edge of $G$). 
By construction, we have  $||E'|| = ||E||$ and $|V'|, |E'|$ are linear in $|V|, |E|$.

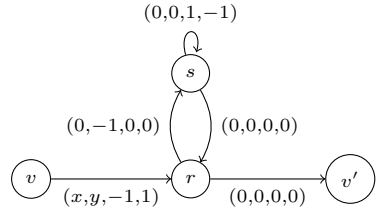
\begin{figure}[h]
\centering
  \begin{tikzpicture}[scale=.7]
      \everymath{\scriptstyle}
		
		\node[draw, circle] (s) at (2,2){$s$};
		\node[draw,circle] (v) at (-1,0){$v$};
		\node[draw,circle] (r) at (2,0){$r$};
		\node[draw, circle] (v') at (5,0){$v'$};
		
		\draw[->] (v) to node [below]{$(x,y,-1,1)$} (r);
		\draw[->] (r) to node [below]{$(0,0,0,0)$} (v');
		\draw[->] (s) edge [loop above] node {$(0,0,1,-1)$} (s);

		\draw[->] (s) to [bend left] node [right]{$(0,0,0,0)$} (r);
		\draw[->] (r) to [bend left] node [left]{$(0,-1,0,0)$} (s);

	\end{tikzpicture}
\caption{Construction of a $4$-dimensional energy game.}
\label{fig:gadget}
\end{figure}

With this reduction, let us prove that the answer to the energy mean-payoff decision problem for $G$ is {\sf Yes} (for both problems $\EMPi^{>0}$ and $\EMPs^{>0}$) if and only if the answer to unknown initial credit problem for $G'$ is {\sf Yes}. Recall that we already know by Corollary~\ref{cor:variants} that the answer is $\sf Yes$ simultaneously for problems $\EMPi^{>0}$ and $\EMPs^{>0}$.

\medskip
Let us first suppose that the answer is {\sf No} for $G'$. Then by Theorem~\ref{thm:multiEnergy}, $\playerTwo$ has a spoiling memoryless strategy $\sigma_2$ that he can use for all initial credits $c'_0 \in \N^4$ in $G'$. As $V_2 = V'_2$, we can interprete $\sigma_2$ in $G$. 
We consider the game structures $G'(\sigma_2)$ and $G(\sigma_2)$ induced by $\sigma_2$ from $G'$ and $G$ respectively, and where  $\playerOne$ is the unique player. By Theorem~\ref{thm:caracterisation}, $\sigma_2$ is winning for $\playerTwo$ in $G$ (for both $\energy(c_0) \cap \MPi(> 0)$ and $\energy(c_0) \cap \MPs(> 0)$, for some $c_0$) if and only if there is no reachable good cycle in $G(\sigma_2)$. Assume the contrary and let $\pi = v_1 v_2 \ldots v_{k+1}$ with $v_{k+1} = v_1$ be such a reachable cycle of length $k$, that is, such that 
\begin{eqnarray} \label{eq:alpha}
w(\pi) =  (\alpha,\beta) \quad  \mbox{ with } \alpha \geq 0 \mbox{ and } \beta >0. 
\end{eqnarray}
We are going to construct from $\pi$ a reachable cycle $\pi'$ in $G'(\sigma_2)$ such that $w'(\pi') \geq (0,0,0,0)$. With Lemma~\ref{lem:energyCycle}, this will contradict $\sigma_2$ being winning for $\playerTwo$ in $G'$.

The cycle $\pi'$ is constructed as follows where for each edge $(v_\ell, v_{\ell + 1})$ of $\pi$ we denote by $r_\ell$ and $s_\ell$ the two new vertices of Figure~\ref{fig:gadget}: 
\begin{itemize}
\item the first edge $(v_1,v_2)$ is replaced the path $v_1 r_1 s_1^{k+1} r_1 v_2$ of length $k+4$ that loops $k$ times in the edge $(s_1,s_1)$ 
\item for each $\ell \in \{2, \ldots, k \}$, the edge $(v_\ell,v_{\ell +1})$ is replaced by the path $v_\ell r_\ell v_{\ell+1}$ of length $2$
 \end{itemize}
By definition of $G'$, this cycle $\pi'$ has a weight $w'(\pi')$ equal to
$$w'(\pi') = (0,-1,0,0) + (0,0,k,-k) + (\alpha,\beta,-k,k).$$
Indeed the sum $(0,-1,0,0) + (0,0,k,-k)$ is the cost of path $r_1 s_1^{k+1} r_1$, and the last term $(\alpha,\beta,-k,k)$ comes from the definition of cost $w'(v_\ell,r_\ell)$ for all $\ell \in \{1, \ldots, k \}$. It follows by (\ref{eq:alpha}) that $w'(\pi') = (\alpha, \beta -1,0,0) \geq (0,0,0,0)$ as announced.

\medskip
Let us now suppose that the answer is {\sf Yes} for $G'$. Then by Theorem~\ref{thm:multiEnergy}, $\playerOne$ has a winning strategy $\sigma'_1$ for some initial credit $c'_0$, that is finite-memory with a memory size $M'$. Let us show how to derive from this strategy a winning strategy $\sigma_1$ for $\playerOne$ in $G$ for some initial credit $c_0$, that is finite-memory and has size $M \leq M'$. In this way the last part of Proposition~\ref{prop:reductionEnergy} will be also proved.

First notice that if a play $\rho' \in \Plays(G')$ is consistent with $\sigma'_1$, then it cannot loop forever on one or on both vertices $r,s$ among the new vertices (see Figure~\ref{fig:gadget}). Otherwise in the first case, $\rho$ would loop on the simple cycle $\pi' = (s,s)$ with weight $w'(\pi') = (0,0,1,-1)$, and in the second case, it would loop on some cycle $\pi' \in \{r,s\}^+$ with at least one occurrence of edge $(r,s)$, thus with a weight vector $w'(\pi')$ such that $w'_2(\pi') < 0$. Hence the energy level of $\rho'$ would not remain above $(0,0,0,0)$ for any initial credit $c'_0 \in \N^d$, which is impossible.


Let us now explain how to construct a strategy $\sigma_1$ in $G$ from the finite-memory winning strategy $\sigma'_1$. Intuitively, the plays $\rho$ consistent with $\sigma_1$ will be derived from plays $\rho'$ consistent with $\sigma'_1$ where we delete factors $\rhoprimfactor{k}{\ell} \in \{r,s\}^+$ such that $r, s$ are the new vertices. We proceed as follows. Let $\pi' u \in V'^* \cdot V$ be a path in $G'$ that is consistent with $\sigma'_1$ and that ends in a vertex $u \in V$. We construct from $\pi' u$ a path $\pi u \in V^* \cdot V$ in $G$ such that each factor of $\pi' u$ of the form $v \lambda v'$ such that $v,v' \in V$ and $\lambda \in (V'_1 \setminus V_1)^+$ is replaced by the factor $v v'$. Notice that each such path $\pi u \in V^* \cdot V$ is derived from a unique path $\pi' u$ that is consistent with $\sigma'_1$. Then when $u \in V_1$, we define $\sigma_1(\pi u)$ as
$$\sigma_1(\pi u) = \sigma'_1(\pi' u).$$
One can check that this strategy $\sigma_1$ is finite-memory with a memory size $M$ less than or equal to the size $M'$ of $\sigma'_1$ (intuitively, in the Moore machine of $\sigma'_1$, we remove the finite portions producing factors $v \lambda v'$ as described above).

It remains to prove that $\sigma_1$ is winning for $\playerOne$ (for both $\energy(c_0) \cap \MPi(> 0)$ and $\energy(c_0) \cap \MPs(> 0)$, for some $c_0$). 
First notice that the energy objective is satisfied because dimension 1 is not affected by the reduction from $G$ to $G'$ (on the first component, only weights $0$ label the new edges, see Figure~\ref{fig:gadget}) and $\sigma'_1$ is winning for the energy objective of dimension 1 in $G'$. So we have to show that each play $\rho$ consistent with $\sigma_1$ satisfies $\MPs_2(\rho) \geq \MPi_2(\rho) > 0$. Consider $\rho$ as a play in the game structure $G(\sigma_1)$ and its cycle decomposition in this structure. As a first step, we show that each (simple) cycle in this decomposition is good.

Take such a cycle $\pi u$, let $k \geq 1$ be its length, and let $\pi' u$ be the cycle of $G'(\sigma'_1)$ from which $\pi u$ is derived. We denote its weight vector by $w(\pi u) = (\alpha, \beta)$ with  $\alpha, \beta \in \Z$. We have to prove that $\alpha \geq 0$ and $\beta > 0$. As $\sigma'_1$ is winning, we have $w'(\pi' u) \geq (0,0,0,0)$ (recall that it is a cycle in $G'(\sigma'_1)$). Moreover
\begin{eqnarray} \label{eq:beta}
w'(\pi' u) = (\alpha, \beta - n, -k + \ell, k - \ell) \geq (0,0,0,0)  
\end{eqnarray}
where $n$ is the number of subpaths deleted from $\pi' u$ to derive $\pi u$ and $\ell$ the total number of edges $(s,s)$, $s \in V' \setminus V$, used by $\pi' u$. It follows from (\ref{eq:beta}) that $\alpha \geq 0$. Moreover,  $k = \ell \geq 1$ (since $k$ is the length of $\pi u$) showing that $n \geq 1$ (as $\ell \geq 1$, at least one subpath has been deleted). Therefore $\beta \geq n > 0$. This shows that the cycle decomposition of $\rho$ in $G(\sigma_1)$ is composed of simple cycles that are all good. 

We can now explain why $\rho$ satisfies $\MPi_2(\rho) > 0$. Let $N = |V| \cdot M$ be the number of vertices of $G(\sigma_1)$. Thus each simple path $\pi$ or cycle $C$ of $G(\sigma_1)$ has a length bounded by $N$, and 
\begin{eqnarray} \label{eq:truc}
w_2(\pi) \geq -N \cdot ||E||, \quad\quad\quad w_2(C) \geq 1 
\end{eqnarray}
(as just explained). Take any prefix $\rhofactor{0}{k}$ of $\rho$ and its cycle decomposition into an acyclic part and $t$ simple cycles. It follows that 
\begin{eqnarray} \label{eq:bidule}
k \leq N + t \cdot N
\end{eqnarray}
and 
$$\begin{array}{lllllll}
w_2(\rhofactor{0}{k}) &\geq& -N \cdot ||E|| + t &&&\mbox{by (\ref{eq:truc})} \\
&\geq& -N \cdot ||E|| + \frac{k}{N} - 1 &&&\mbox{by (\ref{eq:bidule})}. 
\end{array}$$
Therefore the average weight of $\rhofactor{0}{k}$ is at least equal to $\frac{-N \cdot ||E|| + \frac{k}{N} - 1}{k}$ and this lower bound converges to $\frac{1}{N} > 0$. This establishes that $\MPi_2(\rho) > 0$. 

Hence $\sigma_1$ is winning and the proof is completed.
\end{proof}

\subsection{Strategies for $\playerOne$: case of non-strict inequalities}

By Theorem~\ref{thm:caracterisationLarge}, we know that infinite memory may be necessary for $\playerOne$ to win in case of non-strict inequalities. The reduction to multi-dimensional energy games of previous section is thus not applicable for this case. Instead, we show how we can effectively construct a winning strategy for $\playerOne$ by combining an infinite number of finite-memory strategies.

\begin{proposition} \label{prop:infinite}
\label{prop:non-strict-effective}
For both problems $\EMPi^{\geq 0}$ and $\EMPs^{\geq 0}$, if $\playerOne$ is winning from an initial vertex $v_0$, then one can effectively construct a strategy for him to win from $v_0$. This strategy requires infinite memory. 
\end{proposition}

\begin{proof}
Remember by Corollary~\ref{cor:variants} that $\playerOne$ is winning from $v_0$ for the objective $\energy(c_0) \cap \MPi(\geq 0)$ for some $c_0$ if and only if he is winning from $v_0$ for the objective $\energy(c_0) \cap \MPs(\geq 0)$ for some $c_0$. Here, we show how to construct a winning strategy for $\playerOne$ for the mean-payoff-inf case only. Indeed a winning strategy in this case is also winning for the mean-payoff-sup case. 

We first note that if $\playerOne$ is winning from a vertex $v$ for the objective
\begin{equation*}
\Obj(c_0) = \energy_1(c_0) \cap \MPi_2(\geq 0),  
\end{equation*}
\noindent
then he is also winning from $v$ for the objective
\begin{equation*}
\Obj_i(c_0) = \energy_1(c_0) \cap \MPi_2(> -\epsilon_i)
\end{equation*} 
\noindent
for all $\epsilon_i = \frac{1}{2^i}$, $i \in \Nzero$. Let $\sf Win$ be the set of vertices $v$ from which $\playerOne$ is winning for $\Obj(c_0)$ for some $c_0$. In particular $v_0 \in {\sf Win}$ by hypothesis. From now on, we assume that the vertices not in $\sf Win$ are removed from $V$ leading to a game structure that we still denote by $G$. This can be done as a winning strategy for $\playerOne$ will never enter those vertices.

For all vertices $v \in {\sf Win}$, we denote by $c(v) \in \N$ the minimum initial credit from which $\playerOne$ is winning for $\Obj(c(v))$ from $v$. Similarly for all $i \in \Nzero$, we denote by $c_i(v) \in \N$ the minimum initial credit from which he is winning for $\Obj_i(c_i(v))$ from $v$ and by $\sigma^v_i$ such a winning strategy for $\playerOne$. Recall by Proposition~\ref{prop:reductionEnergy} that all strategies $\sigma_i^v$ can be supposed to be finite-memory and to have memory size bounded by $M_i^v$. The game structure $G(\sigma_i^v)$ induced by $\sigma_i^v$ has a number of vertices equal to
\begin{eqnarray} \label{eq:size}
N_i^v = |{\sf Win}| \cdot M_i^v
\end{eqnarray}
Also, we have that $c_1(v) \leq c_2(v) \leq c_3(v) \leq \ldots \leq c(v)$. Moreover as these initial credits are integers, 
\begin{eqnarray} \label{eq:kv}
\exists k_v, \forall i\geq k_v : \quad c_i(v) = c_{k_v}(v).
\end{eqnarray}
Let us define 
\begin{eqnarray} \label{eq:KappaGamma}
\kappa = \max_{v \in \sf Win} k_v \mbox{ and } \gamma = \max \{c_{i+1}(v) - c_i(v) \mid v \in {\sf Win}, i \in \Nzero \}.
\end{eqnarray}
These constants will be useful later for the energy objective.

\subparagraph{An effective winning strategy for $\playerOne$.}
Let us define a strategy $\tau_1$ for $\playerOne$ from $v_0$ that will be proved to be winning for $\playerOne$. A play $\rho$ consistent with $\tau_i$ is the limit of a sequence of prefixes $\rho_i$ of increasing length constructed in the following way:
\begin{enumerate}
\item Initialize $i = 1$ and $\rho_0 = v_0$;
\item Assume that a prefix $\rho_{i-1}$ has been constructed so far and that its last vertex is $v_{i-1}$. Apply, starting from $v_{i-1}$, the strategy $\sigma_i^{v_{i-1}}$ (against $\playerTwo$) until the produced path $\pi_i$ consistent with $\sigma_i^{v_{i-1}}$ and the path $\rho_i$ equal to the concatenation $\rho_{i-1}$ with $\pi_i$ satisfy
\begin{eqnarray} \label{eq:rhoi}
w_2(\rho_i) > N_{i+1}^{v_i} \cdot ||E|| - |\rho_i| \cdot \epsilon_i.
\end{eqnarray}
\item Increment $i$ by 1 and goto 2.
\end{enumerate}

Notice that in (\ref{eq:rhoi}), we require for $w_2(\rho_i)$ more than $w_2(\rho_i) > - |\rho_i| \cdot \epsilon_i$. Indeed the latter inequality would be enough to guarantee that the mean-payoff-sup value of $\rho$ satisfies $\MPs(\rho) \geq 0$ but we will explain later that we need (\ref{eq:rhoi}) to guarantee $\MPi(\rho)\geq 0$.

For the correctness of the given construction, we need to prove that for each $i \in \Nzero$, there exists a path $\rho_i$ satisfying (\ref{eq:rhoi}). This is a consequence of point $(ii)$ of the next lemma.

\begin{lemma}
\label{lem:accumulation}
As each $\sigma_i^v$ is a finite-memory strategy from $v$ winning for $\energy_1(c_0) \cap \MPi_2 (> -\epsilon_i)$,
  \begin{itemize}
   \item[$(i)$] for all plays $\pi$ consistent with $\sigma_i^v$ from $v$, for all $k \in \mathbb{N}$, we have  $w_2(\pifactor{0}{k}) > - N_i^v \cdot ||E|| - k \cdot \epsilon_i$, and
    \item[$(ii)$] for all $K \in \mathbb{N}$, there exists $k \in \mathbb{N}$ such that for all plays $\pi$ consistent with $\sigma_i^v$ from $v$, we have $w_2(\pifactor{0}{k}) > K - k \cdot \epsilon_i$. 
  \end{itemize}
\end{lemma}

\begin{proof}
Let us come back to the game structure $G(\sigma_i^v)$ with $N_i^v$ vertices (by (\ref{eq:size})). As $\sigma_i^v$ is winning for the objective $\MPi_2 (> -\epsilon_i)$, all reachable cycles $C$ in $G(\sigma_i^v)$ have a average weight 
\begin{eqnarray} \label{eq:ccycle}
\frac{w_2(C)}{|C|} > - \epsilon_i.
\end{eqnarray}
Moreover as the weight $w_2(C)$ is an integer, $w_2(C) \geq - |C| \cdot \epsilon_i + t_C$, for some $t_C > 0$. Let $t =  \min\{t_C \mid C \mbox{ reachable cycle in } G(\sigma_i^v) \}$. This tells us that one unit $t> 0$ of weight is accumulated each time a cycle is closed in $G(\sigma_i^v)$:
\begin{eqnarray} \label{eq:cycle+}
w_2(C) \geq - |C| \cdot \epsilon_i + t.
\end{eqnarray}

Let us prove $(i)$. Consider a play $\pi$ consistent with $\sigma_i^v$ from $v$, i.e., an infinite path in $G(\sigma_i^v)$. Let $k \in \mathbb{N}$ and let us reason on the cycle decomposition of $\pifactor{0}{k}$. First, as the acyclic part of this decomposition has a length bounded by $N_i^v$, its weight is bounded below by $-N^v_i \cdot ||E||$. Second, let $\ell$ be the total length of the cycles $C$ of the cyclic decomposition of $\pifactor{0}{k}$. As all cycles $C$ in $G(\sigma_i^v)$ satisfy (\ref{eq:ccycle}), we conclude that the total weight of this cyclic part of $\pifactor{0}{k}$ is bounded below by $- \ell \cdot \epsilon_i$. Finally, as $\ell \leq k$, we obtain the claimed lower bound of $(i)$, that is, $w_2(\pifactor{0}{k}) > - N_i^v \cdot ||E||  - k \cdot \epsilon_i$.

Let us now prove $(ii)$. We simply repeat the arguments given for $(i)$ by using (\ref{eq:cycle+}) instead of (\ref{eq:ccycle}). If $\alpha$ cycles are closed during the cycle decomposition of $\pifactor{0}{k}$, we then get $w_2(\pifactor{0}{k}) \geq \alpha \cdot t - N_i^v \cdot ||E||  - k \cdot \epsilon_i$ instead of the inequality of $(i)$. So, given $K \in \mathbb{N}$, take $k \in \mathbb{N}$ such that $\alpha$ is large enough to get an accumulated positive weight $\alpha \cdot t$ such that $\alpha \cdot t - N_i^v \cdot ||E|| > K$. This establishes $(ii)$.
\end{proof}

Let us prove that $\tau_1$ is a winning strategy (with infinite memory) from $v_0$ for the objective $\Obj(d_0)$ with the initial credit 
\begin{eqnarray} \label{eq:d0}
d_0 = \kappa \cdot \gamma + c_1(v_0)
\end{eqnarray}
with the constants defined in (\ref{eq:KappaGamma}). Let $\rho$ be a play consistent with $\tau_1$ from $v_0$, that is, $\rho$ is the limit of a sequence of prefixes $\rho_i$ as described previously in the definition of $\tau_i$. Remember that each $\rho_i$, $i \in \Nzero$, is the concatenation of $\rho_{i-1}$ and $\pi_i$ such that $\pi_i$ is consistent with $\sigma_i^{v_{i-1}}$ from $v_{i-1}$.

\subparagraph{Mean-payoff-inf objective.}
We begin by showing that $\rho$ satisfies $\MPi_2(\rho) \geq 0$. To achieve this goal, it is enough to show that for all $i \in \Nzero$, the average weight never falls below $- \epsilon_{i-1}$ during the construction of $\rho_i$ (i.e. the construction of $\pi_i$), and this average weight is above $-\epsilon_{i}$ at the end of the construction of $\rho_i$ (see Figure~\ref{fig:MPi}).

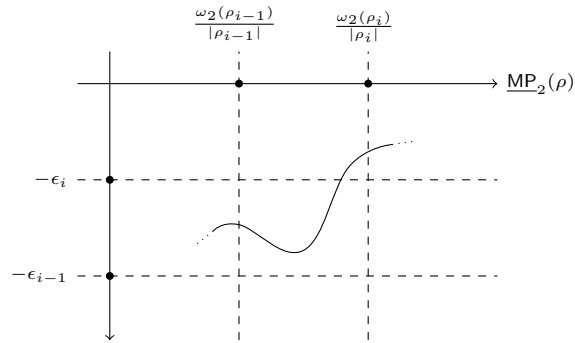
\begin{figure}[h]
\centering
	\begin{tikzpicture}[scale=0.85]
	      \everymath{\scriptstyle}
	
		\draw[dashed] (0.5,2)  -- (1,2) -- (2,2) -- (3,2) -- (4,2) -- (5,2) -- (6,2) -- (7,2);
		\draw[dashed] (0.5,3.5) -- (1,3.5) -- (2,3.5) -- (3,3.5) -- (4,3.5) -- (5,3.5) --(6,3.5) -- (7,3.5) ;
		\draw[->] (0.5,5) -- (1,5) -- (2,5) -- (3,5) -- (4,5) -- (5,5) --(6,5) -- (7,5) ;
	
		\filldraw[]  (0.5,3.5) node[align=center, left] {$-\epsilon_{i}$};
		\filldraw[]  (0.5,2) node[align=center, left] {$-\epsilon_{i-1}$};
		
		\filldraw[]  (1,3.5) circle (1.5pt);
		\filldraw[]  (1,2) circle (1.5pt)  ;
		
		\filldraw[]  (7,5) node[align=center, right] {$ \MPi_2(\rho)$ };
	
		\filldraw[]  (3,5.5) node[align=center, above] {$\frac{\omega_2(\rho_{i-1})}{\mid\rho_{i-1}\mid}$ };
		\filldraw[]  (5,5.5) node[align=center, above] {$\frac{\omega_2(\rho_{i})}{\mid\rho_{i}\mid}$ };
		
		\filldraw[]  (3,5) circle (1.5pt);
		\filldraw[]  (5,5) circle (1.5pt)  ;
	
		\draw [] plot [smooth, tension=0.7] coordinates { (2.6, 2.7) (3, 2.8) (4, 2.4) (4.7,3.7)  (5.37, 4.05)};
		
		\draw [dotted] (5.37, 4.05) -- (5.7, 4.1);
		\draw [dotted] (2.6, 2.7) -- (2.35, 2.5);
		
		\draw[<-]  (1,1) -- (1,2) -- (1,3) --  (1,4) --  (1,5) --  (1,5.5);
	
		\draw[dashed] (3,1)  -- (3,2) -- (3,3) --  (3,4) --  (3,5.5) ;
	
		\draw[dashed] (5,1) -- (5,2) -- (5,3) --  (5,4) --  (5,5.5) ;
	\end{tikzpicture}
\caption{$\rho$ satisfies $\MPi_2(\rho) \geq 0$.}
\label{fig:MPi}
\end{figure}

Let us show that such a property is a consequence of Lemma~\ref{lem:accumulation} and inequality~(\ref{eq:rhoi}) satisfied by $\rho_i$. First by (\ref{eq:rhoi}), the average weight of $\rho_i$ satisfies $\frac{w_2(\rho_i)}{|\rho_i|} > -\epsilon_i$. Second, consider any prefix $\pifactor{0}{k}$ of $\pi_i$ and the corresponding prefix $\rhofactor{0}{k'}$ of $\rho_i$ such that $k' = k+|\rho_{i-1}|$. Then by point $(i)$ of Lemma~\ref{lem:accumulation}, we have $w_2(\pifactor{0}{k}) > - N_i^v \cdot ||E|| - k \cdot \epsilon_i$, and by (\ref{eq:rhoi}) applied to $\rho_{i-1}$, we have $w_2(\rho_{i-1}) > N_{i}^{v_i} \cdot ||E|| - |\rho_{i-1}| \cdot \epsilon_{i-1}$. Therefore we get
\begin{eqnarray*}
w_2(\rhofactor{0}{k'}) &=& w_2(\rho_{i-1}) + w_2(\pifactor{0}{k}) \\
&>& (N_{i}^{v_i} \cdot ||E|| - |\rho_{i-1}| \cdot \epsilon_{i-1}) + (-N_i^v \cdot ||E|| - k \cdot \epsilon_i) \\
&>& - |\rhofactor{0}{k'}| \cdot \epsilon_{i-1}
\end{eqnarray*}
Hence, as announced, we have that the average weight of the prefix $\rhofactor{0}{k'}$ of $\rho_i$ is above $-\epsilon_{i-1}$.

\subparagraph{Energy objective.} It remains to explain why the energy objective is also satisfied by $\rho$ with the initial credit $d_0$ defined in (\ref{eq:d0}). Recall from the definition of $\tau_1$ that $\rho$ is the limit of a sequence of prefixes $\rho_i$ such that each $\rho_i$ is the concatenation of $\rho_{i-1}$ and $\pi_i$. Recall also that $c_i(v) \in \N$ is the minimum initial credit for which $\sigma^v_i$ is winning from $v$.

By construction, $\pi_1$ is consistent with $\sigma_1^{v_0}$ with the initial credit $d_0 = c_1(v_0) + \Delta_1$, where $\Delta_1 = \kappa \cdot \gamma$. Hence the energy level of $\rho_1 = \pi_1$ never drops below zero and it is at least equal to $c_1(v_1) + \Delta_1$ in the last vertex $v_1$ of $\rho_1$ by Lemma~\ref{lem:energy}. Similarly $\pi_2$ is consistent with $\sigma_2^{v_1}$ with the initial credit $c_1(v_1) + \Delta_1 = c_2(v_1) + \Delta_2$, where $\Delta_2 = \kappa \cdot \gamma - (c_2(v_1) - c_1(v_1))$. Hence the energy level of $\rho_2$ never drops below zero and it is at least equal to $c_2(v_2) + \Delta_2$ in the last vertex $v_2$ of $\rho_2$ by Lemma~\ref{lem:energy}. This argument can be repeated for all $i \in \Nzero$: the energy level of $\rho_i$ never drops below zero and it is at least equal to $c_i(v_{i}) + \Delta_i$, with $\Delta_i = \kappa \cdot \gamma - \sum_{j=1}^{i-1} (c_{j+1}(v_{j}) - c_{j}(v_{j}))$. Notice that we always have $\Delta_i \geq 0$ by (\ref{eq:kv}) and by definition of $\kappa$ and $\gamma$ (see (\ref{eq:KappaGamma})). Therefore the energy level of $\rho$ never drops belows zero. 

This proves that $\tau_1$ is a winning strategy for the objective $\energy_1(d_0) \cap \MPi_2(\geq 0)$ and thus conclude the proof.
\end{proof}

\subsection{Proof of Theorem~\ref{thm:twoPlayer}}

We conclude this section with the proof of Theorem~\ref{thm:twoPlayer}.

\begin{proof}[Proof of Theorem~\ref{thm:twoPlayer}]
We establish the three assertions of the theorem as follows.

We first prove that the energy mean-payoff decision problems for two-player games $G$ are in co-NP for the four variants. This result is obtained as follows. By Proposition~\ref{prop:memoryless}, memoryless strategies are sufficient for $\playerTwo$ to win, for all four variants. Hence, the following is an algorithm in co-NP: guess a memoryless strategy $\sigma_2$ for $\playerTwo$, and in the resulting one-player game $G(\sigma_2)$, verify in polynomial time whether $\playerOne$ is winning thanks to Theorem~\ref{thm:onePlayer}.

Second, we consider the two variants with strict inequalities. By Proposition~\ref{prop:reductionEnergy}, there exists a polynomial reduction of the energy mean-payoff decision problem to the unknown initial credit problem for $4$-dimensional energy games. By Theorem~\ref{thm:multiEnergy}, it follows that the energy mean-payoff decision problem can be solved in pseudo-polynomial time and that exponential-memory strategies are sufficient for $\playerOne$ to win.

Finally, we consider the last two variants with non-strict inequalities. In Proposition~\ref{prop:non-strict-effective}, we have shown how we can effectively construct a winning strategy for $\playerOne$ in this case.
\end{proof}

\

\bibliographystyle{abbrv}
\bibliography{biblio}

\end{document}